\documentclass[aps,nopacs,nokeys,superscriptaddress,11pt,twoside]{revtex4}

\usepackage{graphicx,epic,eepic,epsfig,amsmath,latexsym,amssymb,verbatim,revsymb,color}

\usepackage{theorem}
\newtheorem{definition}{Definition}
\newtheorem{proposition}[definition]{Proposition}
\newtheorem{lemma}[definition]{Lemma}

\newtheorem{theorem}[definition]{Theorem}
\newtheorem{corollary}[definition]{Corollary}

\def\squareforqed{\hbox{\rlap{$\sqcap$}$\sqcup$}}
\def\qed{\ifmmode\squareforqed\else{\unskip\nobreak\hfil
\penalty50\hskip1em\null\nobreak\hfil\squareforqed
\parfillskip=0pt\finalhyphendemerits=0\endgraf}\fi}
\def\endenv{\ifmmode\;\else{\unskip\nobreak\hfil
\penalty50\hskip1em\null\nobreak\hfil\;
\parfillskip=0pt\finalhyphendemerits=0\endgraf}\fi}
\newenvironment{proof}{\noindent \textbf{{Proof~} }}{\qed}
\newenvironment{remark}{\noindent \textbf{{Remark~}}}{}

\mathchardef\ordinarycolon\mathcode`\:
\mathcode`\:=\string"8000
\def\vcentcolon{\mathrel{\mathop\ordinarycolon}}
\begingroup \catcode`\:=\active
  \lowercase{\endgroup
  \let :\vcentcolon
  }

\newcommand{\nc}{\newcommand}
\nc{\rnc}{\renewcommand}
\nc{\beq}{\begin{equation}}
\nc{\eeq}{{\end{equation}}}
\nc{\beqa}{\begin{eqnarray}}
\nc{\eeqa}{\end{eqnarray}}
\nc{\lbar}[1]{\overline{#1}}
\nc{\bra}[1]{\langle#1|}
\nc{\ket}[1]{|#1\rangle}
\nc{\ketbra}[2]{|#1\rangle\!\langle#2|}
\nc{\braket}[2]{\langle#1|#2\rangle}
\nc{\proj}[1]{| #1\rangle\!\langle #1 |}
\nc{\avg}[1]{\langle#1\rangle}
\nc{\Rank}{\operatorname{Rank}}
\nc{\smfrac}[2]{\mbox{$\frac{#1}{#2}$}}
\nc{\tr}{\operatorname{Tr}}
\nc{\ox}{\otimes}
\nc{\dg}{\dagger}
\nc{\dn}{\downarrow}
\nc{\cA}{{\cal A}}
\nc{\cB}{{\cal B}}
\nc{\cC}{{\cal C}}
\nc{\cD}{{\cal D}}
\nc{\cE}{{\cal E}}
\nc{\cF}{{\cal F}}
\nc{\cG}{{\cal G}}
\nc{\cH}{{\cal H}}
\nc{\cI}{{\cal I}}
\nc{\cJ}{{\cal J}}
\nc{\cK}{{\cal K}}
\nc{\cL}{{\cal L}}
\nc{\cM}{{\cal M}}
\nc{\cN}{{\cal N}}
\nc{\cO}{{\cal O}}
\nc{\cP}{{\cal P}}
\nc{\cR}{{\cal R}}
\nc{\cS}{{\cal S}}
\nc{\cT}{{\cal T}}
\nc{\cX}{{\cal X}}
\nc{\cZ}{{\cal Z}}
\nc{\csupp}{{\operatorname{csupp}}}
\nc{\qsupp}{{\operatorname{qsupp}}}
\nc{\var}{{\operatorname{var}}}
\nc{\rar}{\rightarrow}
\nc{\lrar}{\longrightarrow}
\nc{\polylog}{{\operatorname{polylog}}}
\nc{\1}{{\openone}}
\nc{\wt}{{\operatorname{wt}}}
\nc{\av}[1]{{\left\langle {#1} \right\rangle}}

\nc{\RR}{{{\mathbb R}}}
\nc{\CC}{{{\mathbb C}}}
\nc{\FF}{{{\mathbb F}}}
\nc{\NN}{{{\mathbb N}}}
\nc{\ZZ}{{{\mathbb Z}}}
\nc{\PP}{{{\mathbb P}}}
\nc{\QQ}{{{\mathbb Q}}}
\nc{\UU}{{{\mathbb U}}}
\nc{\EE}{{{\mathbb E}}}
\nc{\id}{{\operatorname{id}}}

\nc{\CHSH}{{\operatorname{CHSH}}}

\newcommand{\cvartheta}{{\widetilde\vartheta}}
\newcommand{\calpha}{{\widetilde\alpha}}
\newcommand{\Galpha}{{\widehat\alpha}}
\newcommand{\cplus}{{\,\boxplus\,}}

\nc{\be}{\begin{equation}}
\nc{\ee}{{\end{equation}}}
\nc{\bea}{\begin{eqnarray}}
\nc{\eea}{\end{eqnarray}}
\nc{\<}{\langle}
\rnc{\>}{\rangle}
\nc{\Hom}[2]{\mbox{Hom}(\CC^{#1},\CC^{#2})}
\nc{\rU}{\mbox{U}}

\nc{\ob}[1]{#1}

\begin{document}

\title{Zero-error communication via quantum channels, \protect\\
       non-commutative graphs and a quantum Lov\'{a}sz $\mathbf{\vartheta}$ function}

\author{Runyao Duan}
\affiliation{Centre for Quantum Computation and Intelligent Systems (QCIS), Faculty of Engineering \protect\\ and Information Technology, University of Technology, Sydney NSW2007, Australia}
\affiliation{State Key Laboratory of Intelligent Technology and Systems, Tsinghua National Laboratory\protect\\ for Information Science and Technology, Department of Computer Science and Technology,\protect\\ Tsinghua University, Beijing 100084, China}
%\email{runyao@gmail.com}

\author{Simone Severini}
\affiliation{Department of Physics and Astronomy, University College London, WC1E 6BT London, U.K.}
%\email{simoseve@gmail.com}

\author{Andreas Winter}
\affiliation{Department of Mathematics, University of Bristol, Bristol BS8 1TW, U.K.}
\affiliation{Centre for Quantum Technologies, National University of Singapore, 2 Science Drive 3, Singapore 117542}
\email{a.j.winter@bris.ac.uk}

\date{11 March 2010}

\begin{abstract}
We study the quantum channel version of Shannon's zero-error 
capacity problem. Motivated by recent progress on this question,
we propose to consider a certain operator space as the 
quantum generalisation of the adjacency matrix, in terms of
which the plain, quantum and entanglement-assisted capacity can be
formulated, and for which we show some new basic properties.

Most importantly, we define a quantum version of Lov\'{a}sz' famous
$\vartheta$ function, as the norm-completion (or stabilisation)
of a ``naive'' generalisation of $\vartheta$. We go on to show that this
function upper bounds the number of entanglement-assisted
zero-error messages, that it is given by a semidefinite 
programme, whose dual we write down explicitly, and that
it is multiplicative with respect to the natural (strong) graph product.

We explore various other properties of the new quantity, which
reduces to Lov\'{a}sz' original $\vartheta$ in the classical case,
give several applications, and propose to study the
operator spaces associated to channels as ``non-commutative graphs'',
using the language of Hilbert modules.
\end{abstract}

\maketitle

\section{Classical channels, graphs and zero-error communication}
\label{sec:c-intro}
For a classical channel $N:X \rar Y$ between discrete alphabets
$X$ and $Y$ (in the following assumed to be finite), i.e.~a 
probability transition function $N(y|x)$, Shannon~\cite{Shannon56}
initiated the study of zero-error capacities, i.e.~of transmitting
messages by one and asymptotically many uses of the channel.

To transmit messages through this channel with no probability
of confusion, different messages $m$ need to be associated to different
input symbols $x$ in such a way that the output distributions
$N(\cdot|x)$ have disjoint supports. This motivates the 
introduction of the \emph{confusability graph} $G$ of $N$,
that has the vertex set $X$ and an edge $x \sim x'$ whenever $x$ and $x'$ 
can be confused via the channel, i.e.~if there exists
$y \in Y$ such that $N(y|x)N(y|x') \neq 0$. Clearly then, a code
has to consist of an independent set (also known as stable set,
or anti-clique) $X_0 \subset X$,
i.e.~a set of vertices without edges between them. The maximum
size $|X_0|$ of an independent set in $G$ is called the independence
number $\alpha(G)$, and by the preceding discussion it is the
maximum number of messages that can be transmitted through the
channel without the possibility of confusing them.

Using two channels $N_1$ and $N_2$ in parallel means really
that we have a product channel
\[
  N_1 \times N_2: X_1 \times X_2 \rightarrow Y_1 \times Y_2,
  \quad \text{with} \quad
  (N_1 \times N_2)(y_1y_2|x_1x_2) = N_1(y_1|x_1) N_2(y_2|x_2).
\]
If the channels have confusability graphs $G_1$ and $G_2$,
respectively, the confusability graph of the product channel
is the (strong) graph product $G_1 \times G_2$ which
has vertex set $X_1 \times X_2$ and edges
\[
  x_1x_2 \sim x_1'x_2'  \quad \text{iff} \quad 
  \begin{cases}
    \text{ either } x_1 \sim x_1' \text{ and } x_2 \sim x_2', \\
    \text{ or }     x_1 \sim x_1' \text{ and } x_2=x_2', \\
    \text{ or }     x_1 = x_1' \text{ and } x_2 \sim x_2'.
  \end{cases}
\]
(If this looks complicated, it does so because it has to 
encapsulate the idea that a symbol can be confused with itself.)
An integer $n$ uses a channel with confusability graph $G$ is thus
described by the $n$-fold graph product $G^n$.
With this we can define the zero-error capacity of the 
graph as
\[
  C_0(G) = \lim_{n\rar\infty} \frac{1}{n}\log\alpha(G^n)
         = \sup_{n} \frac{1}{n}\log\alpha(G^n),
\]
i.e.~the asymptotically largest number of bits transmissible 
with certainty, per channel use (throughout, $\log$ is understood
as the binary logarithm).
Note that in graph theory the convention is prevalent to
call $\Theta(G) := 2^{C_0(G)} = \sup_n \sqrt[n]{\alpha(G^n)}$
the zero-error capacity, but in this paper we prefer to stay
in keeping with the information theoretic usage.

For some graphs, $C_0(G) = \log \alpha(G)$, but in general the
zero-error capacity is larger -- a well-known example is the
pentagon $C_5$ whose capacity is $\frac{1}{2}\log 5$~\cite{Lovasz}, and 
there are graphs such that for every finite $n$,
$\frac{1}{n}\log\alpha(G^n) < C_0(G)$~\cite{guo}.
Finding $\alpha(G)$ (and a maximal-size independent set)
is in general an NP-hard problem, and the calculation of
the zero-error capacity is even worse as it is not even known whether 
$C_0(G)$ is computable.

It should be noted that Shannon~\cite{Shannon56} also considered
(and solved) the problem of zero-error transmission via many
realisations of $N$ in the presence of instantaneous (passive
and noiseless) feedback. In that case, it is not sufficient to
look at the confusability graph $G$ of the channel, but
rather at the full bipartite graph that represents the
possible input-output transitions. The capacity
$C_{0\text{F}}(N)$ in that case is either $0$, if $C_0(N)=0$,
or given by the logarithm of a
linear programming relaxation of the independence number,
called \emph{fractional packing number}.
Note that $C_{0NS}(N)$, the zero-error capacity in the presence of arbitrary 
non-signalling correlations~\cite{barrett} has the same property,
and in fact is always the logarithm of the fractional packing number~\cite{CLMW}.

A much better upper bound on $\alpha(G)$ was given by 
Lov\'{a}sz~\cite{Lovasz} as a semidefinite programming relaxation,
and called $\vartheta(G)$: rephrasing slightly~\cite[Thms.~5 and 6]{Lovasz},
\begin{equation}
  \label{eq:theta-original}  
  \vartheta(G) = \max \bigl\{ \| \1 + T \| : T_{xx'}=0 \text{ if }x=x\text{ or }x\sim x',
                                             \text{ and } \1+T \geq 0 \bigr\},
\end{equation}
where the maximum is over $|X| \times |X|$ complex
(Hermitian) matrices $T$, though one can show
that it is sufficient to consider \emph{real symmetric} $A$
in the above formula.
In fact, via an expression of $\vartheta$ as
the solution to a semidefinite programme, it can also
be shown to be multiplicative with respect to the graph product 
(i.e.~$\vartheta(G\times H) = \vartheta(G)\vartheta(H)$). Thus,
it also gives an upper bound $C_0(G) \leq \log\vartheta(G)$
on the zero-error capacity.
Apart from some special graphs exhibited by Haemers~\cite{Haemers}
and a particular construction by Alon~\cite{Alon}, it remains
the best upper bound on the zero-error capacity, and has been
deeply studied ever since it appeared~\cite{Knuth}.

\medskip
In the rest of the paper we will extend this theory to quantum
channels and structures generalising the confusability graph 
(see section~\ref{sec:q}).
Instead of introducing only the mathematical objects, we shall
precede each definition by a motivating discussion of the 
zero-error information theory; for instance in section~\ref{sec:0} 
we will introduce zero-error codes for channels to motivate our
definitions of quantum independence numbers (there are at
least three meaningful ones). Then in section~\ref{sec:theta},
we introduce the quantum $\vartheta$ function, explore some of
its properties, of which the most important one is the
semidefinite formulation (section~\ref{sec:sdp}). We end with
highlighting several applications (section~\ref{sec:end}), 
and discussing future directions with non-commutative 
graphs, in section~\ref{sec:q-graphs}, where we propose an
algebraic framework for them.

\section{Quantum channels and non-commutative graphs}
\label{sec:q}
To describe the quantum generalisations of the above combinatorial
concepts, we start with quantum communication channels, mapping
quantum states to quantum states. The input
and output alphabets of a channel are replaced by (complex)
Hilbert spaces $A$ and $B$ -- in the present paper of finite
dimension $|A|$ and $|B|$ -- with their spaces of linear
operators $\mathcal{L}(A)$, etc.
The Hermitian (self-adjoint) operators $\mathcal{L}(A)_{\text{sa}}$
are the physical observables on $A$, while
the states are the \emph{density operators} $\rho \in \mathcal{S}(A) \subset \mathcal{L}(A)$,
i.e.~$\rho \geq 0$ and $\tr\rho = 1$. Note that the set of states is
a convex body whose extreme points are exactly the one-dimensional
projectors $\proj{\psi}$ onto one-dimensional subspaces $\CC\ket{\psi}$
with a unit vector $\ket{\psi} \in A$.
[We use Dirac notation throughout: $\ket{\psi} \in A$ is a vector,
$\bra{\psi}$ it's adjoint (a linear form), $\bra{\varphi} \psi \rangle$
denotes the Hilbert space inner product (linear in the second argument),
and $\ketbra{\psi}{\varphi}$ is the corresponding outer product,
a rank one operator in $\mathcal{L}(A)$, etc.]
A quantum channel is now a linear map
$\mathcal{N}: \mathcal{L}(A) \rightarrow \mathcal{L}(B)$
that is additionally \emph{completely positive and trace preserving
(cptp)}. The latter means that $\tr \mathcal{N}(\rho) = \tr\rho$;
the former means that not only $\mathcal{N}$ maps positive semidefinite
operators into positve semidefinite operators (being a ``positive''
map, for short), but also all extensions $\mathcal{N}\ox\id_{R}$
for an arbitrary Hilbert space $R$. The class of completely
positive maps is the largest subset of positive maps containing
the identity and stable under tensor products~\cite{Effros-book}.

Cptp maps between Hilbert space operator algebras have several
useful representations with associated physical interpretation. One of 
them is the \emph{Kraus form} $\mathcal{N}(\rho) = \sum_j E_j \rho E_j^\dagger$
with Kraus operators $E_j:A \rightarrow B$, which can be read as
the state change under a generalised measurement with ``events'' $j$.
Every such form defines a completely positive map, and it is trace
preserving iff $\sum_j E_j^\dagger E_j = \1$.

Classical channels are embedded into this picture as follows:
starting from the sample space, e.g.~the inputs $X$ to a channel,
we consider the Hilbert space $\CC X$, spanned by the
orthonormal basis $\{ \ket{x} \}_{x\in X}$. The probability
simplex is mapped to the convex hull of the pure basis states
$\proj{x}$, so that we focus only on density operators diagonal
in the computational basis.
%$\{ \ket{x} \}_{x\in X}$. 
A classical channel
$N:X\rar Y$ has to be translated into a cptp map between the
diagonal matrices over $A = \CC X$ and $B=\CC Y$, which is
done canonically by constructing it from the Kraus operators
$\sqrt{N(y|x)}\ketbra{y}{x}$, $x\in X$, $y\in Y$. I.e., for each
classical probabilistic transition $x \rar y$ there is an
event in the quantum cptp map.

For the channel $\mathcal{N}:\mathcal{L}(A) \rightarrow \mathcal{L}(B)$,
with Kraus operators $E_j:A \rightarrow B$, we now define the
\emph{non-commutative (confusability) graph} as the operator subspace
\begin{equation}
  \label{eq:q-graph}
  S := \operatorname{span} \bigl\{ E_j^\dagger E_k : j,k \bigr\} < \mathcal{L}(A).
\end{equation}
In~\cite{CCH-0,Duan-0} it is shown that a subspace $S$ is associated
in the above way to a channel iff $\1 \in S$ and $S=S^\dagger$.
That is why we shall call an operator space $S < \mathcal{L}(A)$
with these properties a \emph{non-commutative graph}, regarding 
the operator space $S$ as the quantum generalization
of the classical confusability graph $G$. This idea is enforced
by the observation that for two channels $\mathcal{N}_1$
and $\mathcal{N}_2$, with associated subspaces $S_1$ and $S_2$,
respectively, the tensor product channel 
$\mathcal{N}_1 \ox \mathcal{N}_2$ has operator subspace
$S_1 \ox S_2$. We shall come back to this notion, with a 
proper (and more subtle) definition, in the last section~\ref{sec:q-graphs}.

Again, let us review this concept in the classical case: as we have
seen, the Kraus operators may be chosen as $E_{xy} = \sqrt{N(y|x)}\ketbra{y}{x}$,
meaning that
\[
  E_{x'y'}^\dagger E_{xy} = \sqrt{N(y'|x')N(y|x)} \bra{y'} y \rangle   \ketbra{x'}{x}
\]
is nonzero iff $y=y'$ and $N(y|x')N(y|x) \neq 0$. Thus, 
\[
  S = \bigl\{ T : \forall x \not\sim x'\  \bra{x} T \ket{x'} = 0 \bigr\},
\]
which means that from the patterns of zeros in the $|X| \times |X|$-matrix
representation of the admissible $T$ we can read off the graph
complement $\overline{G}$ of the confusability graph $G$. 
Note that an operator
space such as this is always a non-commutative graph, and that there is
always a classical channel $N$ giving rise to $S$: simply choose
as the output alphabet $Y$ the set of edges of $G$, and $N$
maps an input symbol to a random edge incident with it.

Coming back to the general case:
An alternative way of thinking about the state change due to
the channel $\mathcal{N}$ is to view it as a pulling-back
of observables on $B$ to observables on $A$: the linear map
effecting this translation is the adjoint
$\mathcal{N}^*(X) = \sum_j E_j^\dagger X E_j$ (in physics
often called the ``Heisenberg picture'', in contrast to the
``Schr\"odinger picture'' $\mathcal{N}$), and indeed
one may think of the channel $\mathcal{N}$ as allowing
the receiver to make (generally distorted) measurements on $A$.
The adjoint map is characterised by being completely
positive and \emph{unital}, i.e. $\mathcal{N}^*(\1) = \1$.

Every channel has a Stinespring dilation, representing the
dynamics as an isometry followed by a partial trace: i.e., there
exists $V:A \hookrightarrow B \ox C$ such that
\[
  \mathcal{N}(\rho) = \tr_C V\rho V^\dagger,
\]
and up to isometric equivalence, $C$ (the ``environment'')
and $V$ are unique. 
Then one has a unique \emph{complementary channel}
\[
  \widehat{\mathcal{N}}(\rho) = \tr_B V\rho V^\dagger,
\]
representing the information loss of the original channel
to the environment. Note that the adjoint maps of $\mathcal{N}$
and $\widehat{\mathcal{N}}$ can be written compactly using the
Stinespring isometry $V$:
\begin{align*}
  \mathcal{N}^*(X)           &= V^\dagger (X \ox \1) V, \\
  \widehat{\mathcal{N}}^*(Y) &= V^\dagger (\1 \ox Y) V.
\end{align*}
recalling that $V^\dagger$ is a projection.

\begin{lemma}
\label{lemma:S-char}
For any channel $\mathcal{N}$ with complementary channel
$\widehat{\mathcal{N}}$,
\(
  S = \widehat{\mathcal{N}}^*\bigl( \mathcal{L}(C) \bigr).
\)
In words: $S$ is the space of operators on $A$ measurable 
by the channel environment.
\end{lemma}
\begin{proof}
We can write a Stinespring dilation of $\mathcal{N}$
via the injective $V:A \hookrightarrow B \otimes C$,
\[
  V\ket{\varphi} = \sum_{j} (E_j\ket{\varphi})^B \otimes \ket{j}^C,
\]
so that for an arbitrary operator $X \in \mathcal{L}(C)$ the
Heisenberg map of the complementary channel reads
\[
  \widehat{\mathcal{N}}^*(X) = V^\dagger(\1 \otimes X)V
                             = \sum_{j,k} E_j^\dagger E_k \bra{j} X \ket{k}.
\]
Now, since the operators $\ketbra{k}{j}$ form a basis of $\mathcal{L}(C)$,
\[
  \widehat{\mathcal{N}}^*\bigl( \mathcal{L}(C) \bigr)
     = \operatorname{span} \bigl\{ E_j^\dagger E_k : j,k \bigr\} = S,
\]
i.e., the image of $\widehat{\mathcal{N}}^*$ is indeed $S$.
\end{proof}

\medskip
Note that this lemma also shows that our definition of $S$ was sound:
it doesn't depend on the particular choice of Kraus operators,
and can be entirely understood in terms of the channel map
(or rather its complement).

\medskip
\begin{remark}
  In general, $S$ does not uniquely define the channel $\mathcal{N}$ from which
  it originates. Already classical graphs and channels
  show this, as the confusability graphs records only which pairs
  of inputs can lead to the same output with the same probability,
  but it doesn't remember the value of this probability,
  nor can it tell us about the triples of inputs which can end
  up at the same output (note that even if there is a triangle
  in $G$, there may not be a single output symbol which can be
  reached by all of its vertices).
\end{remark}

\medskip
Returning to the channel motivation, we can ask what happens to
a non-commutative graph $S = \widehat{\mathcal{N}}^*(\mathcal{L}(C))$
if we add post-processing or pre-processing to the channel
$\mathcal{N}:\mathcal{L}(A) \rar \mathcal{L}(B)$. 
I.e., considering
channels $\mathcal{R}:\mathcal{L}(B) \rar \mathcal{L}(B')$
and $\mathcal{T}:\mathcal{L}(A') \rar \mathcal{L}(A)$, let us look
at the non-commutative graphs $\widehat{S} < \mathcal{L}(A)$ and $S' < \mathcal{L}(A')$
belonging to the compositions
$\mathcal{R}\!\circ\!\mathcal{N}$ and $\mathcal{N}\!\circ\!\mathcal{T}$,
respectively.

Regarding the former, looking at the definition eq.~(\ref{eq:q-graph})
shows that $S < \widehat{S}$, which is the natural relation
of $S$ being a subgraph of $\widehat{S}$.
Regarding the latter, fix a Stinespring isometry $U:A' \hookrightarrow A \ox D$,
and observe that
\[
  S' = \widehat{\mathcal{N}\!\!\circ\!\mathcal{T}}^*\bigl( \mathcal{L}(C\ox D) \bigr)
     = U^\dagger \bigl( S \ox \mathcal{L}(D) \bigr) U.
\]
The projection $U^\dagger: A\ox D \rar A'$ can be understood as giving rise
to an \emph{induced} subgraph (much as a subset of the vertices of
a classical graph would). Note that in this way, every non-commutative
graph $S$ is an induced subgraph of the product $\1_B \ox \mathcal{L}(C)$
of an empty and a complete graph, by virtue of the 
Stinespring dilation $V$ of an appropriate channel $\mathcal{N}$.
We come back to the issue of (induced) subgraphs again in
section~\ref{sec:q-graphs}.

\section{Zero-error communication with and without entanglement}
\label{sec:0}
Zero-error information transmission via general quantum channels was considered
first by Medeiros \emph{et al.}~\cite{Med}, and then by
Beigi and Shor~\cite{BeigiShor} (in those investigations,
communication signals were, implicitly or explicitly, restricted
to product states across multiple channel uses); 
more recently in full generality
by Cubitt, Chen and Harrow~\cite{CCH-0}, Duan~\cite{Duan-0} and 
Cubitt and Smith~\cite{superduper}; 
Duan and Shi~\cite{DuanShi} present results on multi-user quantum
zero-error capacity,
while quantum effects for classical channels were discovered by
Cubitt, Leung, Matthews and Winter~\cite{CLMW}.

Let $\mathcal{N}: \mathcal{L}(A) \rightarrow \mathcal{L}(B)$
be a quantum channel, i.e.~a linear c.p.t.p.~map, with Kraus operators
$E_j:A \rightarrow B$, so that $\mathcal{N}(\rho) = \sum_j E_j \rho E_j^\dagger$.
Then to send messages $m$ one has to associate them with states
$\rho$ such that different states $\rho$, $\sigma$ 
lead to orthogonal channel output states: 
$\mathcal{N}(\rho) \perp \mathcal{N}(\sigma)$, because it is precisely the
orthogonal states that can be distinguished with certainty.
Clearly, these states may, w.l.o.g., be taken as pure, as the
orthogonality is preserved when going to any states in the support
(i.e., the range) of $\rho$, $\sigma$, etc.

Now we make the elementary observation, made in previous work,
that two input pure states $\varphi = \proj{\varphi}$ and $\psi = \proj{\psi}$
for unit vectors $\ket{\varphi}, \ket{\psi} \in A$,
lead to orthogonal output states $\mathcal{N}(\varphi)$
and $\mathcal{N}(\psi)$ iff 
\[
  0 = \tr\mathcal{N}(\varphi)\mathcal{N}(\psi) 
    = \sum_{jk} \bigl| \bra{\varphi} E_j^\dagger E_k \ket{\psi} \bigr|^2,
\]
which says that for all $j,\ k$,
$\bra{\varphi} E_j^\dagger E_k \ket{\psi} = 0$. 
In other words,
\[
  \ketbra{\varphi}{\psi} \perp S = \operatorname{span} \bigl\{ E_j^\dagger E_k : j,k \bigr\},
\]
the non-commutative confusability graph of the channel $\mathcal{N}$, where
the orthogonality is with respect to the 
Hilbert-Schmidt inner product $\tr A^\dagger B$ of operators.

From the above formula it is clear that the maximum number
$\alpha(\mathcal{N})$ of one-shot zero-error distinguishable
messages down the channel is given as the maximum size
of a set of (orthogonal) vectors $\{ \ket{\phi_m} : m=1,\ldots,N \}$
such that
\begin{equation}
  \label{eq:quantum-alpha-characterize}
  \forall m\neq m' \quad \ketbra{\phi_m}{\phi_{m'}} \in S^\perp.
\end{equation}
Since it is only a property of $S$, we shall denote
$\alpha(\mathcal{N})$ also as $\alpha(S)$, and we call it
the \emph{independence number} of $S$. Note that the
defining property of the operators $\ketbra{\phi_m}{\phi_{m'}}$
in eq.~(\ref{eq:quantum-alpha-characterize}) is that they are
rank-one and an orthonormal system orthogonal to $S$,
with respect to the Hilbert-Schmidt inner product.
[In~\cite{BeigiShor} it was proved that computing the
independence number $\alpha(S)$ is QMA-complete, much like
$\alpha(G)$ is known to be NP-complete for graphs.]

There are at least two further reasonable notions of independence number
possible for quantum channels and their confusability graphs.
They are motivated by entanglement-assisted zero-error 
communication, and by the zero-error transmission of quantum
information.

First, to transmit quantum information, one needs a subspace $A'$
of $A$ with projection operator $P$ such that $PSP = \CC P$ -- this
is exactly the Knill-Laflamme error correction
condition~\cite{KnillLaflamme}. For the channel this is precisely the
necessary and sufficient condition for the existence of a decoding cptp map 
$\mathcal{D}:\mathcal{L}(B) \rar \mathcal{L}(A')$ such that the composition
\[
  \mathcal{L}(A') \hookrightarrow \mathcal{L}(A)
                  \stackrel{\mathcal{N}}{\longrightarrow} \mathcal{L}(B)
                  \stackrel{\mathcal{D}}{\longrightarrow} \mathcal{L}(A')
\]
is the identity map.
Let $\alpha_q(S)$ be the largest dimension of such a quantum
error correcting code $A'$, which we call the \emph{quantum
independence number}.

Finally, $\calpha(S)$ is defined to be the largest integer $N$ such that
there exist Hilbert spaces $A_0$ and $B_0$,
a state $\omega \in \mathcal{S}(A_0 \ox B_0)$ and cptp maps
$\mathcal{E}_m:\mathcal{L}(A_0) \rightarrow \mathcal{L}(A)$ 
($m=1,\ldots,N$) such that the $N$ states
$\rho_m = (\mathcal{N}\circ\mathcal{E}_m \ox \id_{B_0})\omega$ are pairwise
orthogonal. This definition of the \emph{entanglement-assisted
independence number} is motivated by the scenario where sender
and receiver share the state $\omega$ beforehand, and the sender
uses the encoding maps $\mathcal{E}_m$ to modulate the state 
before sending her share into the channel. The receiver has to be
able to recover the message from his final state, $\rho_m$.
As before, we can argue that the shared state is w.l.o.g.~pure,
i.e.~$\omega = \proj{\Omega}$ for a unit vector $\ket{\Omega} \in A_0 \ox B_0$,
either by picking $\ket{\Omega}$ from the support of $\omega$,
or by purification. In this way, we can already assume $A_0 \simeq B_0$.
Furthermore, all $\mathcal{E}_m$ have Stinespring dilations
$V_m: A_0 \hookrightarrow A \ox R$ (w.l.o.g.~using the same
extension $R$), so that $\mathcal{E}_m(\rho) = \tr_R V\rho V^\dagger$.
Now it is easily seen that the orthogonality condition on the
$\rho_m$ is equivalent to the the states
\[
  \ket{\phi_m} = (V_m \ox \1)\ket{\Omega} \in A \ox R \ox B_0
\]
forming an independent set for $S \ox \mathcal{L}(R)\mathcal \ox \1_{B_0}$.
We can reformulate this in turn without referring to $B_0$, by noting
that there is a state $\rho \in \mathcal{S}(A \ox R)$ and
unitaries $U_m$ on $A \ox R$ such that
\[
  \tr_{B_0} \ketbra{\phi_m}{\phi_{m'}} = U_m \rho U_{m'}^\dagger,
\]
so that we may regard an entanglement-assisted independent set
as a state $\rho \in \mathcal{S}(A\ox R)$ and a collection of unitaries $U_m$,
such that for all $m \neq m'$, $U_m \rho U_{m'}^\dagger \perp S \ox \mathcal{L}(R)$.

A special case is when the encoding modulation is only unitary,
and the extension system is trivial, $R = \CC$. The largest number
of messages under this additional restriction we denote $\calpha_U(S)$,
and call it the \emph{unitary entanglement-assisted independence number}.

On the other hand, if we lift the restriction that the encoding
maps $\mathcal{E}_m$ have to be trace preserving (but demanding
it for the decoding), we obtain the
\emph{generalised entanglement-assisted independence number}
$\Galpha(S)$: we demand instead that 
$\mathcal{E}_m(\sigma) = \sum_j E_{jm} \sigma E_{jm}^\dagger$
is such that $\sum_j E_{jm}^\dagger E_{jm} \in \text{GL}(A_0)$
is invertible. Since such cp maps still have a Stinespring
dilation, only that $V$ is no longer isometry by invertible,
we arrive at the notion of a \emph{generalised entanglement-assisted
independent set}, consisting as before of a state $\rho \in \mathcal{S}(A \ox R)$
and invertible operators $W_m \in \text{GL}(A\ox R)$, such that
for $m \neq m'$,
$W_m \rho W_{m'}^\dagger \perp S \ox \mathcal{L}(R)$. Note that this
concept even makes sense for \emph{generalised non-commutative
graphs}: $S=S^\dagger$, without the condition that $\1\in S$ but
only assuming that there is a positive definite element $M\in S$.

The following proposition records some elementary properties
of the independence numbers.
\begin{proposition}
  \label{prop:elementary-Watson}
  For all non-commutative graphs $S < \mathcal{L}(A)$,
  \[
    \alpha_q(S) \leq \alpha(S) \leq \calpha_U(S) \leq \calpha(S) \leq \Galpha(S).
  \]
  If $\dim S^\perp < k(k-1)$, then $\alpha(S) < k$; in particular
  $\alpha(S) \leq |A|$.
  Furthermore, $\Galpha(S) \leq 1+\dim S^\perp$, even for generalised
  non-commutative graphs.
  \\
  Finally, all these independence numbers are monotonic
  (non-increasing) under pre- and postprocessing (see section~\ref{sec:q}).
\end{proposition}
\begin{proof}
  The ordering of the five numbers is clear from the definition, 
  and so is the monotonicity.

  For the bound on $\alpha(S)$, note that an independent set requires
  $\ketbra{\phi_m}{\phi_{m'}} \in S^\perp$, for $\1 \leq m\neq m' \leq k$.
  But these operators are clearly mutually orthogonal with respect to the
  Hilbert-Schmidt inner product, hence $k(k-1) \leq \dim S^\perp$.

  For the bound on $\Galpha(S)$, we first present a simple
  argument for $\calpha_U(S)$: Consider an entanglement-assisted
  independent set with trivial $R$: we need a state $\rho \in \mathcal{S}(A)$
  and unitaries $U_1,\ldots,U_N$ such that for all $m\neq m'$,
  $U_m \rho U_{m'}^\dagger \in S^\perp$. 
  In particular, all the operators $U_m\sqrt{\rho}$ are mutually
  orthogonal, implying their linear independence. But then also
  the $U_m \rho U_1^\dagger$ are linearly independent, and they 
  are all in $S^\perp$. Thus, $N-1 \leq \dim S^\perp$.
  
  To bound the generalised independence number, 
  assume that there are $N$ cp maps
  $\mathcal{E}_1,\ldots,\mathcal{E}_N$ from $\mathcal{L}(A_0)$ to $\mathcal{L}(A)$,
  and a state $\rho \in \mathcal{S}(A_0)$.
  In order to write out the condition for a generalised independent set
  explicitly, let us make some additional assumptions. 
  First, we can write $\mathcal{E}_m(\sigma)=\sum_{j=1}^{r(m)} E_{mj} \sigma E_{mj}^\dagger$, 
  where $\{E_{mj}:A_0 \rar A \}_{k=1,\ldots,r(m)}$ is a set of
%  linearly independent 
  Kraus operators for $\mathcal{E}_m$.  
  It is convenient to denote $\mathcal{K}_m := \operatorname{span}\{E_{mj}: 1\leq j\leq r(m) \}$,
  the Kraus operator space of $\mathcal{E}_m$, and $d_m := \dim\mathcal{K}_m$.
  Furthermore,
  $E_m := \sum_j E_{mj}^\dagger E_{mj}$ is a positive definite 
  operator in $\mathcal{L}(A_0)$. We can also assume that $\rho$ is invertible in 
  $\mathcal{L}(A_0)$ as we can always choose $A_0$ to be the support of $\rho$ 
  without changing $N$. Now by a simple calculation,
  the condition for a generalised independent set can be rewritten as follows:
  \[
    E_{mj}\rho E_{m'k}^\dagger \in S^\perp,\quad
    \text{ for }
    1\leq m\neq m'\leq N,\ 1\leq j\leq r(m),\ 1\leq k\leq r(m').
 \]
  Noticing that there is a positive definite operator $M\in S$, this
  implies
  \[
    \tr(M E_{mj}\rho E_{m'k}^\dagger) = 0,\quad
    \text{ for }
    1\leq m\neq m'\leq N,\ 1\leq j\leq r(m),\ 1\leq k\leq r(m').
  \]
  That is, the spaces of linear operators
  $\mathcal{K}_m\sqrt{\rho} = \operatorname{span}\{E_{mj}\sqrt{\rho}: 1\leq j\leq r(m) \}$
  are mutually orthogonal for different $m$, with respect to 
  the generalised inner product given by 
  $\<X,Y\>_{M} = \tr(M X Y^\dagger)$. In particular, the $\mathcal{K}_m$ are 
  linearly independent, and $\mathcal{K}_m \cap \mathcal{K}_{m'} = 0$ for $m \neq m'$.
  Now let us focus on $m'=1$, and w.l.o.g.~assume $d_1 \leq d_m$ for all $m$.
  Then we have
  \[
    E_{mj} \rho E_{1k}^\dagger =: X_{mj,k} \in S^\perp,\quad
    \text{ for }
    2\leq m \leq N,\ 1\leq j\leq r(m),\ 1\leq k\leq r(1).
  \]
  Multiplying both sides of the above from the right by 
  $E_{1k}$, and summing over $k$, we obtain
  \[
    E_{mj} \rho E_1 = X_{mj} := \sum_{k=1}^{r(1)} X_{mj,k} E_{1k} \in S^\perp \mathcal{K}_1,
  \]
  for $2\leq m \leq N$ and $1\leq j \leq r(m)$. Since $\rho$ and $E_1$ are
  invertible, this can be rewritten as
  \[
    E_{mj} \in S^\perp \mathcal{K}_1 (\rho E_1)^{-1},\quad
    \text{ for }
    2 \leq m \leq N,\ 1\leq j\leq r(m),
  \]
  or equivalently 
  \(
    \sum_{m=2}^N \mathcal{K}_m < 
            \operatorname{span}\bigl\{ S^\perp \mathcal{K}_1 \bigr\} (\rho E_1)^{-1}.
  \)
  Noticing that
  \[
    \dim\operatorname{span}\bigl\{ S^\perp \mathcal{K}_1 \bigr\} 
       \leq (\dim S^\perp)(\dim \mathcal{K}_1) = d_1\dim(S^\perp),
  \]
  and $\dim \sum_{m=2}^N \mathcal{K}_m =\sum_{m=2}^N d_m$ 
  (because of linear independence of the $\mathcal{K}_m$), we finally arrive at
  $(N-1) d_1 \leq \sum_{m=2}^N d_m \leq d_1 \dim(S^\perp)$,
  completing the proof.
\end{proof}

\medskip
About the independence numbers $\alpha(S)$, $\alpha_q(S)$ and $\calpha(S)$,
and their associated operational capacities
\begin{align*}
  C_0(S)    &= \lim_{n\rar\infty} \frac{1}{n}\log\alpha\bigl(S^{\ox n}\bigr)
                  \qquad \ \text{[(classical) zero-error capacity]}, \\
  Q_0(S)    &= \lim_{n\rar\infty} \frac{1}{n}\log\alpha_q\bigl(S^{\ox n}\bigr)
                  \qquad \!\text{[quantum zero-error capacity]}, \\
  C_{0E}(S) &= \lim_{n\rar\infty} \frac{1}{n}\log\calpha\bigl(S^{\ox n}\bigr)
                  \qquad \ \text{[entanglement-assisted zero-error capacity]},
\end{align*}
quite a bit is known: In~\cite{Duan-0} examples of $S$ are found
such that $\alpha(S)=1$ but $\alpha(S\ox S) \geq 2$, and examples of
$S_1$ and $S_2$ such that $C_0(S_1)=0$ but $C_0(S_1\ox S_2) \gg C_0(S_2)$;
furthermore, non-commutative graphs $S$ such that $C_0(S) = 0$ but $C_{0E}(S) \geq 1$
(all of which are impossible for classical graphs). In fact, while 
$\alpha(S)$ can be $1$ (and even $C_0(S)=0$) for highly nontrivial graphs
$S$, any non-trivial $S \lvertneqq \mathcal{L}(A)$ (i.e.~not a complete
graph) is easily seen to have $\calpha(S) \geq \calpha_U(S) \geq 2$.
In~\cite{CCH-0} even non-commutative graphs $S_1$ and $S_2$ are shown to
exist such that
$C_0(S_1)=C_0(S_2)=0$, yet $\alpha(S_1\ox S_2) \geq 2$; this result is
further improved in~\cite{superduper} to yield even $\alpha_q(S_1\ox S_2) \geq 2$.

We can similarly define the generalised entanglement-assisted
zero-error capacity
\[
  \widehat{C}_{0E}(S) = \lim_{n\rar\infty} \frac{1}{n}\log\Galpha\bigl(S^{\ox n}\bigr),
\]
which is evidently an upper bound on $C_{0E}(S)$. Both $\Galpha(S)$
and $\widehat{C}_{0E}(S)$ have, by their very definition, an important
symmetry property: for any invertible $W \in \text{GL}(A)$,
\begin{equation}
  \label{eq:congruence}
  \Galpha(S) = \Galpha(W S W^\dagger)\quad\text{and}\quad
  \widehat{C}_{0E}(S) = \widehat{C}_{0E}(W S W^\dagger),
\end{equation}
which are meaningful because $S' = W S W^\dagger$ contains
$WW^\dagger > 0$, and $S'={S'}^\dagger$.
(Note at the same time that $\calpha(S)$ and $C_{0E}(S)$ satisfy 
these equations only if $W$ is a unitary.) We call $S$ and $S'$ as
above \emph{congruent}, $S \simeq S'$,
and denote the congruence class of $S$ by $[S]$.

\medskip
The independence numbers $\alpha_q(S)$, $\alpha(S)$, $\calpha(S)$ and $\Galpha(S)$
are computable: this is obvious for the first two, since they are
formulated in terms of the solvability of a set of real polynomial equations
and inequalities in a finite number of variables.
For the latter two, the potentially unbounded dimension of the entangled
state needed appears to create an issue. However, noting that 
the existence of a zero-error code with $N$ messages can be
cast as real algebraic problem in \emph{non-commuting} variables
with polynomial constraints, we can invoke recent results by 
Pironio \emph{et al.}~\cite{NPA}: these state that a
certain hierarchy of semidefinite programmes asymptotically
characterises the solvability of such constraints by Hilbert
space operators for some sufficiently large dimension.
More precisely, one finds zero-error
codes by solving polynomial equations for increasingly higher 
dimensional entangled states and measurements, and finds 
increasingly better witnesses that certain numbers of messages
cannot be sent with zero-error by climbing higher in the hierarchy.

The algorithms implicit in these remarks are very inefficient (in fact,
we cannot even give an upper bound on the runtime for $\calpha$ and
$\Galpha$), but apart from the QMA-completeness of $\alpha$ 
no results concerning the complexity of the independence numbers have been reported. 
In contrast, as far as we know, none of the asymptotic capacities are even 
known to be decidable -- cf.~\cite{AlonLubetzky}.

\section{A quantum Lov\'{a}sz function}
\label{sec:theta}
For the any non-commutative graph $S < \mathcal{L}(A)$, 
i.e.~$\1\in S$ and $S=S^\dagger$, we make, motivated by
eq.~(\ref{eq:theta-original}), the following definition.
\begin{equation}
  \label{eq:theta-S}
  \vartheta(S) := \max \bigl\{ \| \1+T \| : T \in S^\perp,\ \1+T \geq 0 \bigr\},
\end{equation}
where the norm is the operator norm (i.e.~the largest singular value).
Note that all elements in $S^\perp$ are traceless, hence for
$d=|A|$ the norm on the right hand side is at 
most that of the case where $T$ has $d-1$ eigenvalues $-1$
and a single eigenvalue $d-1$; hence $\vartheta(S) \leq |A|$.

By the discussion in sections~\ref{sec:c-intro} and \ref{sec:q},
for a classical channel $N:X \rightarrow Y$ with confusability
graph $G$, we can model the channel as a cptp map
with Kraus operators $\sqrt{N(y|x)}\ketbra{y}{x}$, 
so $S$ is spanned by all $\ketbra{x'}{x}$ such that $x \sim x'$
or $x=x'$. Thus, the space $S^\perp$
is exactly the set of matrices $T$ with zeros in all entries
$T_{xx'} = 0$ whenever $xx' \in G$ or $x=x'$. Thus, the eligible
$\1+T$ in the definition (\ref{eq:theta-S}) are positive
semidefinite matrices with ones along the diagonal and
zeroes in all entries $(x,x')$ where $x$ and $x'$ are confusable.
The maximum norm in eq.~(\ref{eq:theta-S}) coincides thus
with the expression for $\vartheta(G)$ in~\cite[Thms.~5 and 6]{Lovasz},
and we conclude that $\vartheta(S) = \vartheta(G)$.

The above definition has some desirable properties:
\begin{lemma}
  \label{lemma:theta-alpha}
  For any non-commutative graph $S$, $\alpha(S) \leq \vartheta(S)$. 
  Furthermore, $\vartheta$ is monotonic with respect to
  subgraphs, i.e.~when $S \subset S'$ for two non-commutative graphs,
  then $\vartheta(S) \geq \vartheta(S')$.
\end{lemma}
\begin{proof}
The monotonicity is clear from the definition. For the relation
to $\alpha$, consider a maximal size indendent set
$\{ \ket{\phi_m} : m=1,\ldots,N \}$, i.e.~$N=\alpha(S)$. Then, 
$T = \sum_{m\neq m'} \ketbra{\phi_m}{\phi_{m'}} \in S^\perp$.
Furthermore,
\[
  \1+T \geq \sum_{m} \proj{\phi_m} + \sum_{m\neq m'} \ketbra{\phi_m}{\phi_{m'}} 
       =    \sum_{m, m'} \ketbra{\phi_m}{\phi_{m'}} \geq 0,
\]
so $T$ is eligible in the definition of $\vartheta(S)$. On the other hand,
\[
  \| \1+T \| = \left\| \sum_{m, m'} \ketbra{\phi_m}{\phi_{m'}} \right\| = N,
\]
and we are done.
\end{proof}

\begin{lemma}
  \label{lemma:supermult}
  $\vartheta$ is supermultipicative, i.e.~for non-commutative graphs
  $S_1 < \mathcal{L}(A_1)$ and $S_2 < \mathcal{L}(A_2)$,
  \[
    \vartheta(S_1\ox S_2) \geq \vartheta(S_1)\vartheta(S_2).
  \]
\end{lemma}
\begin{proof}
Observing that the operator
subspace associated to the tensor product $\mathcal{N}_1 \ox \mathcal{N}_2$
of channel $\mathcal{N}_i$ with operator subspaces $S_i$, respectively,
is given by $S_1 \ox S_2$, we can show 

Namely, for $T_i \in S_i^\perp$ and $\1+T_i \geq 0$, we have
\[
  T := T_1\ox\1 + \1\ox T_2 + T_1\ox T_2 
    \in S_1^\perp\ox S_2 + S_1\ox S_2^\perp + S_1^\perp \ox S_2^\perp
    =   (S_1\ox S_2)^\perp,
\]
and $\1+T = (\1+T_1) \ox (\1+T_2) \geq 0$. At the same time
it follows that $\| \1+T \| = \| \1+T_1 \| \, \| \1+T_2 \|$,
and we are done.
\end{proof}

\medskip
It turns out however that, unlike the classical Lov\'{a}sz function, our
definition (\ref{eq:theta-S}) is not multiplicative. In fact,
it fails even for tensoring certain channels with a trivial
channel. (A channel is called \emph{trivial} if it maps all 
states $\rho$ on $A$ to a constant state $\sigma_0$ on $B$,
thus having associated operator subspace $\mathcal{L}(A)$,
which corresponds to the complete non-commutative graph.)

We shall show that one may even take the identity channel
$\id:\mathcal{L}(\CC^d) \rightarrow \mathcal{L}(\CC^d)$,
which has subspace $S = \CC \1_d$. We claim that
\begin{equation}
  \label{eq:identity-x-trivial}
  \sup_n \vartheta\bigl(\1_d \ox \mathcal{L}(\CC^n)\bigr) = d^2.
\end{equation}
\begin{proof}
First we show that the value
$d^2$ can be attained with $n=d$. Namely, for any orthogonal
operator basis of unitaries, including the identity,
$\1=U_0, U_1, \ldots U_{d^2-1}$,
let
\[
  T = \sum_{\alpha=1}^{d^2-1} U_\alpha \ox \overline{U}_\alpha
    = d^2\Phi_d - \1\ox\1,
\]
so that $\1\ox\1+T$ is, up to a normalisation factor of $d^2$,
the maximally entangled state $\Phi_d$.
Since the latter is positive semidefinite, and
\[
  T \in \bigl(\1_d \ox \mathcal{L}(\CC^d)\bigr)^\perp
    =         \1_d^\perp \ox \mathcal{L}(\CC^d),
\]
because all $U_\alpha$ are traceless,
we obtain indeed $\vartheta\bigl(\1_d \ox \mathcal{L}(\CC^d)\bigr) \geq d^2$.

Second, it remains to show that for all $n$,
$\vartheta\bigl(\1_d \ox \mathcal{L}(\CC^n)\bigr) \leq d^2$. For this
consider any $T \in \1_d^\perp \ox \mathcal{L}(\CC^n)$ such that
$\1_d\ox\1_n + T \geq 0$. On the one hand, clearly
$\tr_d (\1_d\ox\1_n + T) = d\1_n$ -- which has norm $d$ --,
on the other hand, it is well-known that the partial trace over
a $d$-dimensional system can change the operator norm (in fact, 
any $p$-norm) by at most a factor of $d$~\cite{Hayden-vanDam}. 
Thus, $\| \1_d\ox\1_n + T \| \leq d^2$. 
\end{proof}

\medskip
This motivates the following better definition, a kind of norm 
completion of $\vartheta$:
\begin{definition}
  \label{defi:c-theta}
  Observing that
  $\bigl( S \ox \mathcal{L}(\CC^n) \bigr)^\perp = S^\perp \ox \mathcal{L}(\CC^n)$,
  let the \emph{quantum Lov\'{a}sz function} be
  \begin{equation}\begin{split}
    \label{eq:c-theta-S}
    \cvartheta(S) &:= \sup_n \vartheta\bigl( S \ox \mathcal{L}(\CC^n) \bigr)               \\
                  & = \sup_n \max \bigl\{ \| \1+T \| : T \in S^\perp\ox\mathcal{L}(\CC^n),
                                                                     \ \1+T \geq 0 \bigr\},
  \end{split}\end{equation}
  where the supremum is over all integers $n$, and the maximum in the
  second line is again over Hermitian operators $T$.
\end{definition}
Note that by our above result on the ideal channel,
also $\cvartheta(S) \leq |A|^2$. And for classical graphs $G$, since
$\vartheta$ is multiplicative and $\mathcal{L}(\CC^n)$ is the
operator space version of the complete graph,
$\cvartheta(S) = \vartheta(S) = \vartheta(G)$.

\medskip
\begin{remark}
From the point of view of operator spaces it might appear rather
natural and pleasing that we have to consider the norm
completion $\cvartheta$, so to speak, of $\vartheta$, by
taking the supremum over tensor products with arbitrary 
full matrix spaces.

There seems to be an analogy to the construction of the completely
bounded norm from the ``naive'' norm of operator maps~\cite{Paulsen-book}.
Much like completely bounded norms~\cite{Watrous}, also our
definition via completion will turn out to be given by a 
semidefinite programme (see the next section).
\end{remark}

\medskip
From the definition it is clear that in general $\cvartheta$ inherits
the supermultiplicativity from $\vartheta$:
\begin{lemma}
  \label{lemma:tilde-theta-supermult}
  $\cvartheta$ is supermultipicative, i.e.~for non-commutative graphs
  $S_1 < \mathcal{L}(A_1)$ and $S_2 < \mathcal{L}(A_2)$,
  \[
    \cvartheta(S_1\ox S_2) \geq \cvartheta(S_1)\cvartheta(S_2).
  \]
  \qed
\end{lemma}

More importantly, however, it is related to the entanglement-assisted
independence number:
\begin{lemma}
  \label{lemma:tilde-theta-alpha}
  For any non-commutative graph $S$, $\calpha(S) \leq \cvartheta(S)$. 
  Furthermore, $\cvartheta$ is monotonic with respect to
  subgraphs, i.e.~when $S \subset S'$ for two non-commutative graphs,
  then $\cvartheta(S) \geq \cvartheta(S')$.
\end{lemma}
\begin{proof} 
The monotonicity is inherited from $\vartheta$. For
the relation to $\calpha$,
the argument is an extension of the one we made for the
unassisted case and $\vartheta(S)$. Namely, recall that we may
pad the channel by a sufficiently large dummy register that
goes into a trivial channel, and find a state $\rho \in \mathcal{S}(A\ox R)$
and unitaries $U_m$ on $A\ox R$ such that
for $1 \leq m \neq m' \leq N$, 
\[
  U_m \rho U_{m'}^\dagger \in S^\perp \ox \mathcal{R}.
\]
Evidently this is unchanged under rescaling $\rho$,
so we replace it by a multiple $X$ with largest eigenvalue $1$:
$X = \proj{\varphi} + X'$,
where $X' \perp \proj{\varphi}$ is a rest which satisfies $\| X' \| \leq 1$.
Now we consider the candidate
\[
  T = \sum_{m\neq m'} U_m X U_{m'}^\dagger \ox \ketbra{m}{m'}
    \in S^\perp \ox \mathcal{L}(R\ox \CC^n).
\]
This is an eligible operator in eq.~(\ref{eq:c-theta-S}) because
\[\begin{split}
  \1 + T &=    \1 + \sum_{m\neq m'} U_m X U_{m'}^\dagger \ox \ketbra{m}{m'} \\
         &\geq \sum_{mm'} U_m X U_{m'}^\dagger \ox \ketbra{m}{m'}           \\
         &=    \left( \sum_{m} U_m\sqrt{X} \ox \ket{m} \right)
                            \left( \sum_{m'} \sqrt{X}U_{m'}^\dagger \ox \bra{m'} \right) 
          =    MM^\dagger \geq 0, 
\end{split}\]
where in the second line we have used $\1 \geq \sum_m U_m X U_m^\dagger \ox \proj{m}$.

Finally, to bound the norm, define the unit vector
$\ket{\phi} = \frac{1}{\sqrt{N}}\sum_m U_m\ket{\varphi} \ox \ket{m}$.
Then observe
\[
  \| \1 + T \| \geq \| MM^\dagger \|
               \geq \bra{\phi} M M^\dagger \ket{\phi}
               =    N,
\]
which completes the proof.
\end{proof}

%\medskip\noindent
%Motivated by eq.~(\ref{eq:congruence}), we thus define
%\(
%  \cvartheta([S]) := \inf_{S' \simeq S} \cvartheta(S').
%\)

\section{Semidefinite formulation and other properties}
\label{sec:sdp}
We shall now simplify the expression for $\cvartheta$,
putting an a priori limit on the dimension $n$ of the extension system.
Namely, for fixed $n$, we have
\[
   \| \1\ox\1+T \| =  \max_{\ket{\phi}}\ \bra{\phi} (\1\ox\1+T) \ket{\phi},
\]
for $T \in S^\perp\ox\mathcal{L}(\CC^n)$, $\1\ox\1+T \geq 0$, where
the maximum is over unit vectors in $A\ox\CC^n$.
Now we can use a trick analogous to Lov\'{a}sz'~\cite[Theorem 4]{Lovasz}:
with the maximally entangled vector $\ket{\Phi} = \sum_{i=1}^{|A|} \ket{i}^A\ket{i}^{A'}$
there exists an operator $M:A' \rar \CC^n$ such that
$\ket{\phi} = (\1\ox M)\ket{\Phi}$. Thanks to $\tr_A \Phi = \1_{A'}$,
the normalisation of $\ket{\phi}$ translates into $\rho = M^\dagger M$ being
a state (i.e. of trace $1$) on $A'$.
Thus,
\[
  \bra{\phi} (\1\ox\1+T) \ket{\phi} 
      = \bra{\Phi} \bigl( \1\ox\rho + (\1\ox M^\dagger)T(\1\ox M) \bigr) \ket{\Phi},
\]
and the crucial observation is that
$T' = (\1\ox M^\dagger)T(\1\ox M) \in S^\perp \ox \mathcal{L}(A')$.
As a consequence, we have proved
\begin{theorem}
\label{thm:c-theta-SDP}
For any non-commutative graph $S < \mathcal{L}(A)$,
\begin{align}
  \label{eq:c-theta-SDP}
  \cvartheta(S) &= \max\ \bra{\Phi} (\1\ox\rho + T') \ket{\Phi} \\
          & \phantom{==}\text{s.t. }\ T' \in S^\perp\ox\mathcal{L}(A'), \quad \tr\rho = 1, \nonumber \\
          & \phantom{==\text{s.t. }}\ \1\ox\rho+T' \geq 0, \quad \rho \geq 0, \nonumber
\end{align}
which is a semidefinite characterisation of $\cvartheta$. \qed
\end{theorem}

This has two important consequences: first, we have now an optimisation
with a \emph{bounded} dimension of the extension (namely $|A|$)
and furthermore it is semidefinite~\cite{SDP}, so it is computable efficiently. 
Second, and much deeper,
we have a dual semidefinite programme for the same value that is a minimisation 
problem and allows us to put upper bounds on $\cvartheta(S)$.

\begin{theorem}
\label{thm:c-theta-dual-SDP}
The dual of the semidefinite programme (\ref{eq:c-theta-SDP}) gives
\begin{align}
  \label{eq:c-theta-SDP-dual}
  \cvartheta(S) &= \min\ \| \tr_A Y \| \\
                & \phantom{==}\text{s.t. }\ Y \in S \ox \mathcal{L}(A'), \quad
                                            Y \geq \Phi, \nonumber
\end{align}
where $A'$ is isomorphic to $A$.
\end{theorem}

%\medskip
\noindent
Before we prove this, we record an immediate corollary:
\begin{corollary}
  \label{cor:c-theta-multi}
  $\cvartheta$ is multiplicative: for non-commutative graphs
  $S_1 < \mathcal{L}(A_1)$ and $S_2 < \mathcal{L}(A_2)$,
  \begin{equation}
    \label{eq:c-theta-multi}
    \cvartheta(S_1 \ox S_2) = \cvartheta(S_1) \cvartheta(S_2).
  \end{equation}
\end{corollary}
Indeed, we know already that it is supermultiplicative, so we only have
to show that it is also submultiplicative. I.e., for two subspaces
$S_i < \mathcal{L}(A_i)$,
$\cvartheta(S_1\ox S_2) \leq \cvartheta(S_1)\cvartheta(S_2)$.
But that we can read off from the dual: if $Y_1$ is dual feasible
for $S_1$ and $Y_2$ for $S_2$, then clearly $Y_1 \ox Y_2$ is
dual feasible for $S_1 \ox S_2$. At the same time,
\(
  \| \tr_{A_1A_2} Y_1 \ox Y_2 \| = \| (\tr_{A_1} Y_1) \ox (\tr_{A_2} Y_2) \|
                                 = \| (\tr_{A_1} Y_1) \| \, \| (\tr_{A_2} Y_2) \|,
\)
and we are done. \qed

\medskip\noindent
\begin{proof}{\bf (of Theorem~\ref{thm:c-theta-dual-SDP})}
The primal is a semidefinite programme of the general form
\[
  \max\ \tr CX \ \text{ s.t. }\ \ell(X)=\underline{b},\ X \geq 0,
\]
with a linear vector-valued function $\ell:\mathcal{L}(H)_{\text{sa}} \rar \RR^n$.
The dual of such a form is given by
\[
  \min\ \underline{b}^\top\cdot\underline{y} \ \text{ s.t. }\  \ell'(\underline{y}) \geq C,
\]
where $\ell':\RR^n \rar \mathcal{L}(H)_{\text{sa}}$ is the adjoint linear map
to $\ell$~\cite{SDP}.

In the present case, let $d=|A|$; the matrix $X$ will be
\[
  X = \left[\begin{array}{cc}
        X_{11} & * \\
          *  & X_{22}
      \end{array}\right]
    = \left[\begin{array}{cc}
        \rho & * \\
          *  & \1\ox\rho+T'
      \end{array}\right],
\]
and the linear constraint has to ensure this form of the matrix, the
trace normalisation of $X_{11}$ and the fact that $T'$ is
orthogonal to $S\ox\mathcal{L}(A')$. The objective function is given by 
$C = \left[\begin{array}{cc} 0 & 0 \\ 0 & \Phi \end{array}\right]$.

Thus, fixing an operator basis $\{F_\alpha\}_\alpha$ of $S$, and a 
basis $\{G_\beta\}_\beta$ of $\mathcal{L}(A')$,
the components of $\ell$ are
\begin{align*}
  \ell_0(X)             &= \tr X_{11} 
                         = \tr X\left[\begin{array}{cc} \1 & 0 \\ 0 & 0 \end{array}\right]
                         =: \tr X L_0, \\
  \ell_{\alpha\beta}(X) &= \tr(F_\alpha \ox G_\beta)(X_{22}-\1\ox X_{11})
                         = \tr X\left[\begin{array}{cc}
                                  (-\tr F_\alpha)G_\beta & 0 \\ 
                                                       0 & F_\alpha \ox G_\beta 
                                \end{array}\right]
                         =: \tr X L_{\alpha\beta},
\end{align*}
while $b_0=1$ and all other $b_{\alpha\beta} = 0$.
With these notations, the adjoint map $\ell'$ can be constructed as
\[
  \ell'(\underline{y}) = y_0 L_0 + \sum_{\alpha\beta} y_{\alpha\beta} L_{\alpha\beta}.
\]
Using that the second term in $\ell'$ is a generic element of
$S\ox\mathcal{L}(A')$, we can simplify our expressions, and find 
that the objective function is $y_0$, and that
\[\begin{split}
  \ell'(\underline{y}) &= \left[\begin{array}{cc} Y_{11} & 0 \\ 0 & Y_{22} \end{array}\right], 
                                                                              \ \text{ where}\\
                       &\phantom{==} Y_{22} \in S\ox\mathcal{L}(A'),                         \\
                       &\phantom{==} Y_{11} = y_0\1 - \tr_A Y_{22}.
\end{split}\]
I.e., the constraints are $Y_{22} \geq \Phi$ and $y_0\1 \geq \tr_A Y_{22}$,
proving the form of the dual, since the optimal $y_0$ is the norm
(maximum eigenvalue) of $\tr_A Y_{22}$.

To finish, we only need to verify feasibility of both
primal and dual; for the primal this is shown by $T'=0$, 
for the dual by $Y=d\1\ox\1$. Thus, the conditions of strong
duality are fulfilled, both primal and dual optimal 
values are finite and equal.
\end{proof}

\medskip
Using this dual, we can now show that $\cvartheta$ is monotonic under
pre- and post-processings of the channel defining $S$, and more generally
under enlarging the graph and going to induced subgraphs.
\begin{corollary}
  \label{cor:tilde-theta-mono}
  For non-commutative graphs $S < \widehat{S}$, 
  $\cvartheta(S) \geq \cvartheta(\widehat{S})$.
  
  For a non-commutative graph $S$, let $U:A_0 \hookrightarrow A$
  be an isometry and consider the induced subgraph
  $S' = U^\dagger S U < \mathcal{L}(A_0)$. Then, $\cvartheta(S') \leq \cvartheta(S)$.
  
  As a consequence,
  let $S = \widehat{\mathcal{N}}^*(\mathcal{L}(C))$ with a channel
  $\mathcal{N}:\mathcal{L}(A) \rar \mathcal{L}(B)$. 
  Then, $\cvartheta$ is non-increasing when going to non-commutative
  graphs obtained by either pre- or post-processing $\mathcal{N}$.
\end{corollary}
\begin{proof}
  For a larger graph $\widehat{S} > S$,
  we know already $\cvartheta(\widehat{S}) \leq \cvartheta(S)$; and
  post-processing gives exactly rise to a larger graph $\widehat{S} > S$.
  
  Induced subgraphs are more interesting: 
  Let $Y \in S \ox \mathcal{L}(A')$ be an optimal solution of the dual semidefinite
  programme for $\cvartheta(S)$, according to Theorem~\ref{thm:c-theta-dual-SDP},
  i.e.~$Y \geq \Phi_{AA'}$ and $\| \tr_A Y \| = \cvartheta(S)$.
  But then $Y' = (U^\dagger\ox\1) Y (U\ox\1) \in S'\ox\mathcal{L}(A')$,
  $Y' \geq \Phi_{A_0A_0'}$ and $\| \tr_{A_0} Y' \| \leq \| \tr_A Y \|$.
  
  Finally, any graph $S' < \mathcal{L}(A')$ originating from a pre-processing
  of $\mathcal{N}$ is obtained from an isometry $U:A' \hookrightarrow A \ox D$, 
  via $S' = U^\dagger (S\ox\mathcal{L}(D)) U$. But $S$ and $S \ox \mathcal{L}(D)$
  have the same $\cvartheta$, by definition, and since $S'$ is an induced subgraph
  of the latter, we are done.
\end{proof}

\medskip
We end this section by remarking that the dual in Theorem~\ref{thm:c-theta-dual-SDP}
simplifies considerably in the case of classical channels,
i.e.~$S = \operatorname{span}\{ \ketbra{x}{x'} : x=x' \text{ or } x\sim x' \} < \mathcal{L}(\CC X)$.
Note that $\ket{\Phi} \in A \ox A'$ is invariant under unitaries
of the form $U\ox\overline{U}$, and that $S$ is stabilised by diagonal
unitaries $Z = \sum_x e^{i\varphi_x} \proj{x}$. Hence, with
every dual feasible $Y$, we get an equally good dual feasible
solution $(Z\ox\overline{Z}) Y (Z\ox\overline{Z})^\dagger$, so
by the triangle inequality, we can find a dual optimal solution
among the operators invariant under conjugation with $Z\ox\overline{Z}$,
i.e.~$Y = \sum_{xx'} Y_{xx'} \ketbra{xx}{x'x'}$. The constraints
are $Y \geq \Phi = \sum_{xx'} \ketbra{xx}{x'x'}$ and $Y_{xx'} = 0$
if $x \not\sim x'$, while the objective function is the
norm of the partial trace $\tr_A Y = \sum_x Y_{xx} \proj{x}$.
Thus, we arrive at
\begin{corollary}
  \label{cor:G-dual}
  For a classical graph $G$, Lov\'{a}sz'
  $\vartheta$ is given by the semidefinite programme
  \[
    \vartheta(G) = \min \left\{ \max_{x\in X} Y_{xx} : Y \in S,\ Y \geq J \right\},
  \]
  where $S$ is the non-commutative graph associated to $G$,
  meaning $Y_{xx'}=0$ whenever $x \not\sim x'$, and $J$ is the all-$1$ matrix.
  \qed
\end{corollary}

\section{Applications and discussion}
\label{sec:end}
There are a few immediate consequences,
the most obviously important being
obtained by putting together Lemma~\ref{lemma:tilde-theta-alpha}
and Corollary~\ref{cor:c-theta-multi}:
\begin{corollary}
  \label{cor:C_E-upper}
  For any non-commutative graph $S < \mathcal{L}(A)$,
  $C_{0E}(S) \leq \log\cvartheta(S)$.
  %
%  In fact, since the generalised entanglement-assisted
%  zero-error capacity is the same for all
%  non-commutative graphs congruent to $S$, even
%  $C_{0E}(S) \leq \widehat{C}_{0E}(S) \leq \log\cvartheta([S])$.
  \qed
\end{corollary}

\medskip
Then, for a classical channel with confusability
graph $G$, we observed earlier that $\cvartheta(S) = \vartheta(G)$.
Hence, $\calpha(G) \leq \vartheta(G)$ and so:
\begin{corollary}
  \label{cor:C0-E-theta}
  For any graph $G$,
  $C_{0E}(G) \leq \log\vartheta(G)$.
  \qed
\end{corollary}
This answers an open question from~\cite{CLMW}, which
is nontrivial because there it is shown that $\calpha(G)$
may be strictly larger than $\alpha(G)$.

E.g., we can now compute the entanglement-assisted zero-error capacity of
the ``Bell-Kochen-Specker'' channels discussed in~\cite{CLMW}.
These are all disjoint unions of $n$ copies of
$K_d$, with some extra edges between the complete components,
such that $G$ is exactly the orthogonality graph of a set
of $nd$ vectors in $\CC^d$. If the set of vectors gives rise to
a Kochen-Specker proof of non-contextuality, this means
$\alpha(G) \leq n-1$. On the other hand, in~\cite{CLMW} it is
shown that $\calpha(G) \geq n$, using a rank-$d$ maximally
entangled state $\frac{1}{d}\Phi_d$. Here, we can now see
$n \leq \calpha(G) \leq \vartheta(G) \leq n$, as shown by the
dual feasible solution $Y = n \bigoplus \Phi_d$, which has
$\| \tr_A Y \| = n$. Thus, $\calpha(G) = \vartheta(G) = n$
and we also learn that $C_{0E}(G) = \log n$.
(One could, however, see this also directly by noting that
$G$ contains a disjoint union of $n$ complete graphs as a
subgraph.)

While we do not have a separating upper bound for the unassisted
capacity $C_0(G)$ of these graphs, of course even
as a bound on the independence number, our Corollary~\ref{cor:C_E-upper}
is an improvement over Lov\'{a}sz~\cite{Lovasz}, since we
find that $\vartheta(G)$ is even larger or equal than $\calpha$.
In this sense, the increase of independence number from
$\alpha$ to $\calpha$ due to entanglement-assistance
somehow ``explains'' the fact that Lov\'{a}sz' $\vartheta$ 
is not always a tight bound~\cite{Haemers} -- and in fact, it is quite 
possible that $C_{0E}(G)$ can be strictly larger than $C_0(G)$. 

There are, furthermore,
other quantum channels for which $\cvartheta(S) = \vartheta(S)$.
For instance, perhaps the simplest one is $S=\Delta^\perp$, where
$S^\perp = \CC\Delta$ (with a traceless Hermitian operator $\Delta$)
is one-dimensional. In that case one has evidently 
$\calpha_U(S) = \calpha(S) = \Galpha(S) = 2$
(see Proposition~\ref{prop:elementary-Watson}).
In fact, the gap between $\Galpha(S)$ and $\cvartheta(S)$
can be made arbitrarily large, since the latter can be up to $d$
as shown by the example of
\[
  \Delta = \left[ \begin{array}{cccc}
                     d-1 &    &        &   \\
                         & -1 &        &   \\
                         &    & \ddots &   \\
                         &    &        & -1
                  \end{array}\right],
\]
and in fact similar examples show that every real value between $2$ and $d$ is
realised as some $\cvartheta(\Delta^\perp)$. 
We do not know a better upper bound on $C_{0E}(S)$ for this channel
other than $\log\cvartheta(S)$.

%At the same time, Corollary~\ref{cor:C_E-upper} allows us to determine that
%$C_{0E}(S) = \widehat{C}_{0E}(S) = 1$, because we can explicitly use the congruence
%relation to transform $S$ into $S'={\Delta'}^\perp$ with traceless
%$\Delta' \geq -\1$ and $\cvartheta(S') \leq \| \1+\Delta' \| = 2+o(1)$.
%In particular, we learn also that $\Galpha(S) = 2$.
%
%\medskip
%This calculation begs the immediate question whether similar improvements
%to the upper bound on the zero-error capacity can be made to the
%Lov\'{a}sz $\vartheta$ function of classical channels. We do not know
%if that is possible, but it is not inconceivable as $\vartheta(G)$
%is known not to be tight~\cite{Alon,Haemers}.

\medskip\noindent
Non-commutative graphs for which we can
determine $C_{0E}$ include all $S < \mathcal{L}(\CC^2)$:
\begin{itemize}
  \item If $S = \CC\1$, then the channel is perfect, and by superdense
    coding we can achieve $\calpha(S) = 4 = \cvartheta(S)$,
    hence $C_{0E}(S) = 2$.
  \item The other extreme is $S=\mathcal{L}(\CC^2)$, then $\cvartheta(S)=1$
    and hence $C_{0E}=0$.
  \item In the intermediate case, $2 \leq \dim S \leq 3$, and we claim
    $C_{0E}(S) = 1$. Indeed, the capacity is largest for the smallest
    subspace, hence we consider only $\dim S = 2$. The subspace is spanned
    by $\1$ and another operator, which we may take to be diagonal
    and traceless, thus w.l.o.g.~$Z$. This is the subspace corresponding
    to the noiseless classical (i.e.~$Z$-dephasing) channel
    $\mathcal{N}(\rho) = \sum_{b=0,1} \proj{b} \rho \proj{b}$,
    which clearly has entanglement-assisted capacity $1$, even
    in the Shannon setting~\cite{C_E}, which can be achieved
    error-free and without entanglement since $\alpha(S)=2$, $C_0(S)=1$.
    For $\dim S = 3$, we still have $\calpha(S) = \calpha_U(S) = 2$, by 
    Proposition \ref{prop:elementary-Watson}.
\end{itemize}

Yet another one can be found in~\cite[Thm.~3, eq.~(8)]{Duan-0}, where 
a channel is constructed with non-commutative
graph $S=\1_2\ox\1_d + \1^\perp\ox\mathcal{L}(\CC^d)$, so that
$S^\perp = \1_2\ox\1_d^\perp$. It was shown that $\calpha(S) \geq d^2$,
and indeed, because $S$ contains $\mathcal{L}(\CC^2)\ox\1_d$,
$\cvartheta(S) \leq \cvartheta(\mathcal{L}(\CC^2)\ox\1_d) = d^2$,
hence $C_{0E}(S) = 2\log d$.

Perhaps the most interesting open question regarding the 
entanglement-assisted zero-error capacity is whether
$C_{0E}(S) = \log \cvartheta(S)$. 
Note that this would imply that $C_{0E}$ is 
multiplicative (whereas $C_0$ is not~\cite{Alon});
one might recall that entanglement-assistance has made also
the theory of communication via quantum channels more elegant~\cite{C_E},
and likewise so-called \emph{XOR games}~\cite{XOR-product},
for which a semidefinite characterisation lead to multiplicativity
of the optimal winning probability.
A most challenging test case is presented by the above non-commutative
graphs $S = \Delta^\perp$, for which we do not even know 
$\calpha(S\ox S)$ at the time of writing, nor in fact $\Galpha(S\ox S)$.
Does it perhaps hold that 
$\widehat{C}_{0E}(S) \leq \log\bigl( 1+\dim S^\perp \bigr)$ in general?
-- which by the above examples would imply a separation between $\log\cvartheta(S)$
and $C_{0E}(S)$.

Another question pertains to a possible generalisation of a property
of Lov\'{a}sz' $\vartheta(G)$: Is it true that
$\cvartheta(S_1\cap S_2) \leq \cvartheta(S_1)\cvartheta(S_2)$?
Note that it holds for classical graphs -- because the intersection is
an induced subgraph of the strong product along the diagonal --, and
that it would be an extension of the multiplicativity statement.

Third, it is a bit unsatisfactory that we have three entanglement-assisted
independence numbers. Are they really different?
Is it perhaps true that at least they lead to the same asymptotic capacities?
%Can one at least extend the dimension upper bound in
%Proposition~\ref{prop:elementary-Watson} to $\calpha$?
What is in general the relation between $\Galpha(S)$ and $\cvartheta(S)$?

Finally, looking back at our path, it may seem odd and in
fact a bit arbitrary that we arrived at a Lov\'{a}sz type bound on
the entanglement-assisted independence number. Do there
exist similar bounds for the unassisted zero-error
capacity and the zero-error quantum capacity that are strictly
better than $\cvartheta(S)$?

\section{Non-commutative graph theory?}
\label{sec:q-graphs}
In this last section, no longer concerned with zero-error communication
but driven by the idea of developing a proper theory of non-commutative
graphs, we will finally give the proper definition of graphs, of
subgraphs, induced substructures, etc.
For this purpose, we have to come back to the characterisation of
$S$ in terms of the adjoint $\widehat{\mathcal{N}}^*$ of the complementary
channel (Lemma~\ref{lemma:S-char}).
Such maps, by being completely positive and unital,
obey the Kadison-Schwarz (operator) inequality
\[
  \widehat{\mathcal{N}}^*(X)^\dagger \widehat{\mathcal{N}}^*(X) \leq \widehat{\mathcal{N}}^*(X^\dagger X),
\]
for all $X\in \mathcal{L}(C)$. The set of operators which satisfy this
with equality is, by Choi's theorem~\cite{choi:multi,Paulsen-book}, 
the so-called \emph{multiplicative domain}
\begin{equation}
  \label{eq:multi-domain}
  \mathcal{M} := \bigl\{ X\in \mathcal{L}(C) \text{ s.t. } \forall Y\ 
        \widehat{\mathcal{N}}^*(X)\widehat{\mathcal{N}}^*(Y) = \widehat{\mathcal{N}}^*(XY) \bigr\},
\end{equation}
which is in fact a $*$-subalgebra (containing $\1$)
of $\mathcal{L}(C)$, and restricted to it,
$\widehat{\mathcal{N}}^*: \mathcal{M} \rightarrow S_0 := \widehat{\mathcal{N}}^*(\mathcal{M})$
is a $*$-algebra homomorphism. The image $S_0$ is clearly a subspace of $S$,
a $*$-algebra itself, and by eq.~(\ref{eq:multi-domain}) it satisfies
\[
  S_0 S = S S_0 = S,
\]
i.e., $S$ is a (left and right) $S_0$-module, all presented explicitly as
operator subspaces of $\mathcal{L}(A)$. In fact, it is even a so-called
\emph{Hilbert-$S_0$-module}~\cite{Lance}; all we need is to choose an
$S_0$-valued inner product $\langle \cdot,\cdot \rangle: S\times S \rar S_0$,
which we shall however always assume to be defined on $\mathcal{L}(A)$.
The inner product should be linear in the first, and conjugate linear in the second
element, $\langle X,Y \rangle^* = \langle Y,X \rangle$,
it should respect the module structure (from the right)
as $\langle X,Ya \rangle = \langle X,Y \rangle a$ for $X,Y \in S$ and $a\in S_0$
(which is equivalent to $\langle Xa,Y \rangle = a^\dagger \langle X,Y \rangle$),
and $\langle X,X \rangle \geq 0$ with equality iff $X=0$. This defines
a very strong notion of orthogonality in $\mathcal{L}(A)$.

%conditional expectation $E:\mathcal{L}(A) \rightarrow S_0$ (that is, a completely
%positive and unital linear map that is surjective and has the property
%that $E(XY)=XE(Y)$ and $E(YX)=E(Y)X$ for $X \in S_0$. 
%
%--- Does the channel $\mathcal{N}$ give rise 
%    to this conditional expectation as well?? ---
% NO!!!!!!!!
%%%%%%%%%%%%%%%%%%%%%%%%%%%%%%%%%%%%%%%%%%%%%%%%%%%%%%%%%%%%%%%%%%%%%%%%%%%%%%%%
%Note that if $S_0$ is multiplicity free, that is a direct sum
%of full matrix algebras (where all the tensor-factor identities
%are one-dimensional), $E$ is unique: with projectors $P_j$ such
%that $\sum_j P_j = \1$,
%\(
%  E(X) = \sum_j P_j X P_j.
%\)

In the first part of the paper, we effectively
treated every non-commutative graph as if it had trivial
$S_0=\CC\1$. In this case, there is a whole family of inner products
$\langle X,Y \rangle = (\tr X^\dagger R Y S) \1$ 
for some positive definite $0 < R,S \in\mathcal{L}(A)_{\text{sa}}$,
but there are many more.
Because of its importance for the independent set question discussed
above, and its relation to matrix multiplication, we assign special
status to the Hilbert-Schmidt inner product (i.e.~$R=S=\1$), 
which shall be the default when no inner product is specified.

To obtain some more structure, note that since, $S_0 < S$ it is
reasonable to demand that $\langle \1,Y \rangle = Y$ and
$\langle X,\1 \rangle = X^\dagger$ for $X,Y\in S_0$ (which is
equivalent to asking $\langle X,Y \rangle = X^\dagger Y$).
Motivated by this, and using also the left module structure,
we could ask for the even stronger property
$\langle X,aY \rangle = \langle a^\dagger X,Y \rangle$ for $a \in S_0$
(together with $\langle \1,\1 \rangle = \1$).

In general, the structure theorem for finite dimensional
$*$-algebras implies
\[
  S_0 = \bigoplus_{j=1}^r \mathcal{L}(A_j) \ox \1_{Z_j},
  \quad \text{with} \quad
  A = \bigoplus_{j=1}^r A_j \ox Z_j,
\]
while $\langle \cdot,\cdot \rangle$ gives rise to a conditional
expectation $E(X) = \langle \1,X \rangle$ (satisfying
$E(XA) = E(X)A$ for $X\in S$ and $A\in S_0$ and $E(A) = A$, by the above
additional assumption).
The general form of the conditional expectation is
\[
  E(X) = \bigoplus_{j=1}^r 
            \bigl[ \tr_{Z_j} (P_j\ox\sqrt{\zeta_j})X(P_j\ox\sqrt{\zeta_j}) \bigr] \ox \1_{Z_j},
\]
with the projectors $P_j=\1_{A_j}$ onto $A_j$ and $Q_j=\1_{Z_j}$ 
onto $Z_j$, and states $\zeta_j \in \mathcal{S}(Z_j)$. 
For the conditional expectation to be \emph{faithful},
i.e.~$X\geq 0$ and $E(X)=0$ implying $X=0$, it is necessary and
sufficient that all the $\zeta_j$ are faithful.

Now, from the left and right module structure,
\begin{equation}\begin{split}
  \label{eq:S-structure}
  S                           = \bigoplus_{j,k=1}^r &(P_j\ox Q_j) S (P_k\ox Q_k),\ 
                                                        \text{ and for each } j,\ k, \\
  &(P_j\ox Q_j) S (P_k\ox Q_k) = (\mathcal{L}(A_j)\ox Q_j) S (\mathcal{L}(A_k)\ox Q_k) \\
  &\phantom{(P_j\ox Q_j) S (P_k\ox Q_k)}         = \mathcal{L}(A_k\rar A_j) \ox S_{jk},
\end{split}\end{equation}
where $S_{jk} < \mathcal{L}(Z_k\rar Z_j)$ such that $S_{kj} = S_{jk}^\dagger$
and $Q_j \in S_{jj}$. From this we see that each non-commutative graph gives
rise to an underlying classical graph ``skeleton''
\[
  G(S_0 < S) := \bigl( V=[r], E={jk : S_{jk} \neq 0} \bigr).
\]

A general inner product is not uniquely defined by its conditional
expectation, but each conditional expectation $E$ gives rise to the
following canonical inner product
\begin{equation}
  \label{eq:E-inner-prod}
  \langle X,Y \rangle_E = E(X^\dagger Y)
                        = \bigoplus_{j=1}^r 
            \bigl[ \tr_{Z_j} (P_j\ox\sqrt{\zeta_j})X^\dagger Y(P_j\ox\sqrt{\zeta_j}) \bigr] \ox \1_{Z_j}.
\end{equation}
As before for the
Hilbert-Schmidt inner product, the tracial states $\zeta_j = \frac{1}{|Z_j|}\1_{Z_j}$
are distinguished because of the symmetry of the resulting inner product
\begin{equation*}
%  \label{eq:symmetry}
  E(UXU^\dagger) = U E(X) U^\dagger,\ \text{ for unitaries } U \text{ s.t. } US_0U^\dagger = S_0,
\end{equation*}
which characterises them uniquely.
(And hence its relation to the usual matrix product.) This choice is 
understood as the default if we only specify $S_0$ but not an inner product.

\medskip
We did not need all this additional
structure before, but it motivates our eventual definition:

\begin{definition}
  \label{def:q-graph}
  A \emph{non-commutative graph} is a pair $S_0 < S$ of operator
  subspaces of some $\mathcal{L}(A)$, with a complex Hilbert space $A$,
  equipped with an inner product $\langle X,Y \rangle$
  that makes $S$ a Hilbert left and right $S_0$-module.
  
  That is, $S_0$ is a $*$-subalgebra of $\mathcal{L}(A)$ containing
  $\1$ and contained in $S$,
  $S = S^\dagger$ and $S$ is a left and right $S_0$-module with respect to
  matrix multiplication, i.e.~$S_0 S = S S_0 = S$. 
  The inner product satisfies $\langle X,Y \rangle^* = \langle Y,X \rangle$
  for all $X,Y \in \mathcal{L}(A)$, $\langle \1,\1 \rangle = \1$,
  $\langle X,Ya \rangle = \langle X,Y \rangle a$ for $X,Y \in S$ and $a\in S_0$
  (which is equivalent to $\langle Xa,Y \rangle = a^\dagger \langle X,Y \rangle$),
  $\langle X,aY \rangle = \langle a^\dagger X,Y \rangle$, for $a\in S_0$,
  and $\langle X,X \rangle \geq 0$, with equality iff $X=0$.

  In fact, with the conditional expectation $E(X) = \langle \1,X \rangle$
  from $\mathcal{L}(A)$ to $S_0$, we shall only look at inner products
  of the form $\langle X,Y \rangle = E(X^\dagger Y)$.
  
  To emphasise the dependence on $S_0$ and $E$, we shall call $S$
  a \emph{(non-commutative) $S_0$-graph} (if we don't specify the
  inner product), or more precisely an \emph{$E$-graph}.
\end{definition}
We do not have a sufficient overview over the literature to claim that this
concept is entirely new and unexplored. The term -- apart from a single occurrence 
in the context of non-commutative geometry~\cite{Filk} -- appears not to
have been used before. And while there is some literature regarding
finitely generated Hilbert-modules over finite-dimensional $*$-algebras,
to the best of our knowledge no-one seems to ever have made the connection 
to graph theory.

\medskip
\begin{remark}
Abstractly, there seems no reason to insist on $\mathcal{N}$ being
trace preserving, which would correspond to $S$ not necessarily containing
the identity matrix $\1$. In the first part of the paper we could indeed
have relaxed the definition of non-commutative graph to be an
operator space $S=S^\dagger < \mathcal{L}(A)$ containing some
positive definite element $D > 0$. 

However, the above concepts do not go well with this generalisations,
as for non-unital $\widehat{N}^*$ we do not have unital $*$-subalgebra
structure of $S_0$, nor is it characterised by Choi's theorem. 
Thus we stick with our original definition for now,
leaving an exploration of alternative definitions for later.
\end{remark}

\medskip
Clearly, the same operator subspace $S$ can originate from different
channels, which might however have different $S_0$. All pairs $S_0 < S$
according to the above definition occur, however. 
%--- PROVE IT!
The $*$-subalgebra $S_0$ serves as a kind of ``diagonal'' in the
operator space $S$, in fact, whereas $S$ generalises the edges of
a graph, $S_0$ is representative of the vertices (their number
being remembered in the dimension $|A|$ of the underlying Hilbert space).
It is perhaps helpful to remember, for the sake of intuition, to
recall one of the original motivations to consider Hilbert 
modules~\cite{Kaplansky} as an abstract version of vector bundles
over manifolds, represented as the module of vector fields over the
algebra of continuous functions, where the inner product originates
from a Riemannian structure of the vector bundle; this intuition
has been immensely fruitful in the creation of non-commutative geometry
and its applications~\cite{Connes}.

The basic example of course is once more the classical graph: we saw
before that starting from a noisy channel, one can arrive at the
confusability graph in its non-commutative guise
\[
  S = \operatorname{span}\{ \ketbra{x}{x'} : x=x' \text{ or } x\sim x' \}
    < \mathcal{L}(\CC X),
\]
but that made no distinction between vertices ($x=x'$) and proper
edges. Looking at the quantum version of the channel as discussed
in section~\ref{sec:q}, one can see that $S_0$ will contain all $\proj{x}$.
By appropriately modifying the channel $N$, for instance by
considering $N' = \frac{1}{2}N \oplus \frac{1}{2}\id_X$ (with the
same input alphabet $X$ and the larger output alphabet $X \stackrel{\cdot}{\cup} Y$),
one can indeed enforce $S_0 = \operatorname{span}\{ \proj{x} : x \in X \}$,
with the canonical conditional expectation $\operatorname{diag}(X) = \sum_x \proj{x} X \proj{x}$.
Now, it is clear that one can recover the graph $G$ from $S_0 < S$
up to isomorphism. Thus, the classical graphs are precisely
the $\operatorname{diag}$-graphs, the graph structure recovered
precisely as the skeleton $G(S_0<S)$ of the algebraic data.

Furthermore, the module $S$ over $\operatorname{diag}$ is generated
by a single element (using left and right multiplication), for
instance by the Laplacian of the graph. This property is shared
by all non-commutative graphs where in eq.~(\ref{eq:S-structure}),
$S_{jk}$ is at most one-dimensional for all $j$, $k$.

A class of examples that are already more ``quantum'' are graphs
$S < \mathcal{L}(A)$,
with $S_0 = \CC\1$ and a conditional expectation of the form 
$E_\rho(X) = \tr(\rho X)\1$ for a state $\rho \in \mathcal{S}(A)$.
Such a non-commutative graph we call \emph{$\rho$-graph}, and all
we require for it is $\1\in S=S^\dagger$.

\medskip
Going back once more to the motivation of our concepts from channels,
one may recall that each classical channel $N:X\rar Y$ also gives rise to a
\emph{bipartite} graph with vertex set $X \stackrel{\cup}{.} Y$,
where $x\in X$ and $y\in Y$ are connected by an edge iff $N(y|x) > 0$.
This bipartite graph captures much more about the channel than the
confusability graph, and indeed Shannon's zero-error feedback
result~\cite{Shannon56} and Cubitt \emph{et al.}'s regarding assistance
by non-signalling resources~\cite{CLMW} can be formulated in terms
of this bipartite graph. As its quantum version we propose to 
consider, for a quantum channel $\mathcal{N}:\mathcal{L}(A) \rar \mathcal{L}(B)$
with $\mathcal{N}(\rho) = \sum_j E_j \rho E_j^\dagger$, the
operator subspace
\begin{equation}
  \label{eq:bipartite-graph}
  Z := \operatorname{span}\{ E_j \} < \mathcal{L}(A\rar B).
\end{equation}
This space was crucial in the proof of Proposition~\ref{prop:elementary-Watson},
and it is evident that the non-commutative graph of the channel 
is obtained as $S = Z^\dagger Z < \mathcal{L}(A)$. Furthermore, one
can confirm that indeed $Z$ is still a right $S_0$-module, and that
the above conditional expectation $E$ makes it indeed a Hilbert
module, via the same rule $\langle X,Y \rangle = E(X^\dagger Y)$
for $X,Y \in Z$.

\begin{proof}
  This is essentially only an extension of Choi's reasoning~\cite{choi:multi}.
  We use the Stinespring representation of the channel, with isometry
  $V:A \hookrightarrow B \ox C$, such that 
  $\widehat{\mathcal{N}}^*(m) = V^\dagger (\1_B \ox m_C) V$.
  
  Then, a generic element of $Z$ can be written
  $X = (\1_B\ox\bra{\xi}_C)V = \sum_j \xi_j E_j$ for an appropriate
  vector $\ket{\xi} \in C$. A generic element of $S_0$ instead
  is $a = \widehat{\mathcal{N}}^*(m)$, for an element $m\in\mathcal{M} < \mathcal{L}(C)$
  of the multiplicative domain. We wish to show that $Xa \in Z$,
  and indeed we will find that
  \[
    Xa = (\1_B\ox\bra{\xi'}_C)V,\ \text{ with } \ket{\xi'} = m^\dagger \ket{\xi}.
  \]
  First, noting $Xa = (\1_B\ox\bra{\xi}_C)V V^\dagger (\1_B \ox m_C) V$, 
  Choi's theorem tells us
  \[\begin{split}
    a^\dagger X^\dagger X a &= V^\dagger (\1_B \ox m_C^\dagger) V
                                  V^\dagger (\1_B\ox\proj{\xi}_C) V
                              V^\dagger (\1_B \ox m_C) V                       \\
                            &= V^\dagger (\1_B\ox m^\dagger \proj{\xi}_C m) V
                             = V^\dagger (\1_B\ox\proj{\xi'}_C) V,
  \end{split}\]
  hence $Xa = (U^{X,a}_B\ox\bra{\xi'}_C) V$ 
  for some unitary $U^{X,a} \in \mathcal{U}(B)$.
  But for another $Y \in Z$, $b\in S_0$, once more by Choi's theorem,
  \[\begin{split}
    b^\dagger Y^\dagger X a &= V^\dagger (\1_B \ox n_C^\dagger) V
                                  V^\dagger (\1_B\ox\ketbra{\upsilon}{\xi}_C) V
                              V^\dagger (\1_B \ox m_C) V                       \\
                            &= V^\dagger (\1_B\ox n^\dagger \ketbra{\upsilon}{\xi}_C m) V,
  \end{split}\]
  showing $U^{X,a} = U^{Y,b}$ for all $X,Y$ and $a,b$, which
  concludes the proof.
\end{proof}

\medskip\noindent
This motivates the following definition, for which each cptp map
yields an example:
\begin{definition}
  \label{defi:bipartite-graph}
  A \emph{non-commutative (directed) bipartite graph} with
  ``vertex spaces'' $A$ and $B$ is a subspace 
  $Z < \mathcal{L}(A \rar B)$ together with a unital $*$-subalgebra
  $S_0 < S = Z^\dagger Z < \mathcal{L}(A)$, and a conditional
  expectation $E:\mathcal{L}(A) \rar S_0$, such that
  $Z$ is a right $S_0$-module, and indeed a Hilbert module
  for the inner product $\langle X,Y \rangle = E(X^\dagger Y)$.
\end{definition}
Again, all non-commutative bipartite graphs originate from some cptp channel.
%%%% --- PROVE IT!! --- %%%%

\medskip
We call an $E$-graph $S < \mathcal{L}(A)$ and an $E'$-graph $S' < \mathcal{L}(A')$
\emph{isomorphic}, if there exists a unitary
isomorphism $U$ between $A$ and $A'$ such that
\[
  U S U^\dagger = S',\quad
  U S_0 U^\dagger = S_0',\quad\text{and}\quad
  U E(X) U^\dagger = E'(U X U^\dagger).
\]
This implies a definition of automorphism, too, and we denote the
automorphism group of $S$ as $\text{Aut}(S) < \mathcal{U}(A)$.
We say that the automorphism group \emph{acts (vertex) transitively}
if the only operators in the commutant of $S_0$ that also commute with
the automorphism group, are $\CC\1$.

\bigskip
Now we can start defining the usual graph notions: we call
a \emph{complete graph} a pair $S_0 < S = \mathcal{L}(A)$ together
with any faithful conditional expectation $E:\mathcal{L}(A) \rar S_0$.
To be precise, it is the complete $E$-graph.

The \emph{complement} of an $E$-graph $S$ is defined to be the subspace
\[
  S_0 < S^c := S_0 + S^{(\perp_E)} 
             = S_0 + \{ X\in\mathcal{L}(A) : \forall Y\in S\ \langle X,Y \rangle_E = 0 \},
\]
which by virtue of the Hilbert-module property is again an $E$-graph.
The definition is made in such a way that $(S^c)^c = S$ and $S \cap S^c = S_0$; 
in particular the complete $E$-graph is the complement of $S_0$
(which we call the \emph{empty} $E$-graph).
Note that the notion of complement depends on the conditional expectation
$E$ and its image $S_0$.

\medskip
Also \emph{graph products} are defined easily: and $E$-graph $S < \mathcal{L}(A)$
(with subalgebra $S_0$) and an $E'$-graph $S' < \mathcal{L}(A')$
(with subalgebra $S_0'$) give rise to the (strong) product, which
is the $E\ox E'$-graph $S\ox S' < \mathcal{L}(A\ox A')$, with
subalgebra $S_0 \ox S_0'$. Thus we also have the powers $S^{\ox n}$,
which are $E^{\ox n}$-graphs.

The \emph{disjoint union} of an $E$-graph $S < \mathcal{L}(A)$
and an $E'$-graph $S' < \mathcal{L}(A')$ is the direct sum
$S \oplus S' < \mathcal{L}(A\oplus A')$ (with subalgebra
$S_0 \oplus S_0'$). Denoting the projections onto $A$ and $A'$
in $A\oplus A'$ by $P$ and $P'=\1-P$, this is an $E\oplus E'$-graph,
where $(E\oplus E')(X) := E(PXP) \oplus E'(P'XP')$.
If the graphs originate from channels, their direct sum originates from
the direct sum channel. Note that the corresponding orthogonal sectors
in the direct sum are always perfectly distinguishable; one can
make them indistinguishable by adding the full operator sets
$\mathcal{L}(A\rar A')$ and $\mathcal{L}(A'\rar A)$ to the direct sum,
``filling up the off-diagonal blocks'':
\[
  S \cplus S' := S \oplus S' + \mathcal{L}(A\rar A') + \mathcal{L}(A'\rar A)
               < \mathcal{L}(A\oplus A'),
\]
which we call the \emph{complete union} of the graphs (because it
corresponds to placing a complete bipartite graph between the
vertex spaces $A$ and $A'$).

Clearly, products, disjoint and complete unions are associative,
and both unions are distributive with respect to the graph product.

\begin{proposition}
  \label{prop:additive}
  Both $\vartheta$ and $\cvartheta$ are additive under disjoint unions:
  \[
    \vartheta(S \oplus S') = \vartheta(S) + \vartheta(S'),
    \quad
    \cvartheta(S \oplus S') = \cvartheta(S) + \cvartheta(S').
  \]
  Furthermore, $\alpha$ is additive, and $\calpha$ and $\Galpha$ are 
  superadditive under disjoint unions:
  \[
    \alpha(S \oplus S') = \alpha(S) + \alpha(S'),
    \quad
    \calpha(S \oplus S') \geq \calpha(S) + \calpha(S'),
    \quad
    \Galpha(S \oplus S') \geq \Galpha(S) + \Galpha(S').
  \]  
  Finally, all $f \in \{ \alpha, \calpha, \Galpha, \vartheta, \cvartheta \}$ 
  satisfy the following identity:
  \[
    f(S \cplus S') = \max\bigl\{ f(S), f(S') \bigr\}.
  \]
\end{proposition}
\begin{proof}
  We only need to show the first claim for $\vartheta$. By eq.~(\ref{eq:theta-S}),
  \[
    \vartheta(S \oplus S') = \max \left\{ \left\| \left[ \begin{array}{rr}
                                                       \1+T      & M \\
                                                       M^\dagger & \1+T'
                                                     \end{array} \right] \right\| : 
                                                     T \in S^\perp,\ T' \in {S'}^\perp \right\},
  \]
  where the maximum is restricted to positive semidefinite block matrices.
  It is an easy observation that for all positive semidefinite $L_1$, $L_2$,
  \[
    \max_M \left\{ \left\| \left[ \begin{array}{ll}
                                    L       & M \\
                                    M^\dagger & L'
                                  \end{array} \right] \right\|
                                : \left[ \begin{array}{ll}
                                    L       & M \\
                                    M^\dagger & L'
                                  \end{array} \right] \geq 0 \right\} 
        = \| L \| + \| L' \|,
  \]
  from which the assertion follows.
  
  The independence numbers are clearly superadditive, so it is left to show
  that $\alpha(S \oplus S') \leq \alpha(S) + \alpha(S')$.
%  and $\Galpha(S \oplus S') \leq \Galpha(S) + \Galpha(S')$. 
  For this, let $\{ \ket{\phi_m}: m=1,\ldots,N \}$ be an independent set of
  $S \oplus S'$, so that for all $m \neq m'$,
  \[
    (S \oplus S')^\perp \ni \ketbra{\phi_m}{\phi_{m'}}
                          = (P\oplus P') \ketbra{\phi_m}{\phi_{m'}} (P\oplus P'),
  \]
  which is equivalent to
  \[
    P\ketbra{\phi_m}{\phi_{m'}}P \in S^\perp
    \quad\text{and}\quad
    P'\ketbra{\phi_m}{\phi_{m'}}P' \in {S'}^\perp.
  \]
  Thus, up to normalisation, the set
  $\cA = \{ m : P\ket{\phi_m} \neq 0 \}$ gives rise to an independent
  set in $S$, and likewise $\cB = \{ m : P'\ket{\phi_m} \neq 0 \}$ for $S'$.
  Because each $m$ is in at least one of $\cA$ or $\cB$, the
  claim follows.
   
  Finally, the $\boxplus$-$\max$-identities follow almost immediately from
  $(S \boxplus S')^\perp = S^\perp \oplus {S'}^\perp$.
\end{proof}

\bigskip
Another easy notion is the \emph{distance-$\leq\!\! t$-graph} of an $E$-graph
$S < \mathcal{L}(A)$: this is the subspace $S^t = S\cdot S \cdots S$
(the t-fold product), which is indeed an $E$-graph. By convention,
here $S^0 := S_0$.

\medskip
It may happen that the same $S$ is an $E$-graph (with subalgebra $S_0$)
and an $F$-graph (with subalgebra $S_1 > S_0$), such that $E$ factors
through $F$, i.e.~there is a conditional expectation $G:S_1 \rar S_0$
such that $E = G \circ F$. We call then the $F$-graph $S_1 < S$
a \emph{refinement} of the $E$-graph $S_0 < S$. (The idea being
that with $F$ and $S_1$, the graph has more vertices.) Conversely,
by concatenating the conditional expectation $E$ with another
one $E':S_0 \rar S_1 < S_0$, we can obtain \emph{coarse grainings}
of any $E$-graph as $E'\circ E$-graphs.

The notions of subgraph and induced subgraph are more subtle, because
we have to take care of the conditional expectation. The simplest is
when $S'$ is a \emph{proper subgraph} of an $E$-graph $S$, which means
that $S_0 < S' < S$ and that $S'$ is a sub-Hilbert-$S_0$-module of $S$
with the same inner product: $S_0 S' = S' S_0 = S'$. We call proper
subgraphs also \emph{$E$-subgraphs}.
Less strict, we call an $E'$-graph $S'$ with subalgebra $S_0'$ a 
(generally: improper) \emph{subgraph} of the $E$-graph $S < \mathcal{L}(A)$
if $S' < S$ and $S_0' < S_0$, and $E'|_{S'} = E|_{S'}$

Induced subgraphs of an $E$-graph $S < \mathcal{L}(A)$ (with algebra $S_0$)
are defined with respect to a subspace $A' < A$ with projector $P$: 
the $E'$-graph $S' := PSP < \mathcal{L}(A')$ (with algebra $S_0'$)
is called \emph{proper induced subgraph} if $PS_0P=S_0'$ and the
restriction to $A'$ commutes with the conditional expectations:
$E'(PXP) = P E(X) P$ for all $X\in \mathcal{L}(A)$.
Again, there is a less strict notion of \emph{induced subgraph},
which only demands $PS_0P < S_0'$ and that $S_0' < S'$ is a refinement
of $PS_0P < S'$.

\medskip
To illustrate these notions, we note that the Stinespring dilation theorem
implies that every $E$-graph $S_0 < S < \mathcal{L}(A)$ is a
proper induced subgraph of the strong product between a complete
$F$-graph $\mathcal{L}(C)$ an empty $\rho$-graph $\CC\1 < \mathcal{L}(B)$.
We can also re-interpret the independence numbers of a
non-commutative graph $S$ as the largest 
dimensions of (improper) induced subgraphs: an induced
empty $\rho$-graph $\CC\1 < \mathcal{L}(A')$ for $\alpha_q(S)$
[and such $A'$ we should hence call a \emph{quantum independent set}],
and an induced $\operatorname{diag}$-graph 
$\operatorname{span}\{ \proj{m} : m=1,\ldots,|A'| \} < \mathcal{L}(A')$
for $\alpha(S)$ [and such $A'$ we should call an \emph{independent set}].
Cliques are defined analogously.

Going back to eq.~(\ref{eq:S-structure}), recall that 
a non-commutative graph $S$ has the form
\[
  \bigoplus_{j=1}^r \mathcal{L}(A_j)\ox\1_{Z_j} = S_0 
           < S = \bigoplus_{j,k=1}^r \mathcal{L}(A_k\rar A_j) \ox S_{jk}.
\]
From this we can construct the graph $\widetilde{S}_0$,
\[
  \bigoplus_{j=1}^r \CC\1_{Z_j} =: \widetilde{S}_0 
           < \widetilde{S} := \bigoplus_{j,k=1}^r S_{jk},
\]
which is an induced subgraph of $S$ over a \emph{commutative}
diagonal $\widetilde{S}_0$. In addition, $S$ itself is an induced subgraph
of $\widetilde{S}\ox\mathcal{L}(R)$ for large enough $|R|$. One
can think of $S$ as obtained from $\widetilde{S}$ by ``blowing
up the vertices'': each vertex becomes a complete graph $K_{|A_j|}$,
and each edge a complete bipartite graph $K_{|A_k|,|A_j|}$.
Because of these relations, $S$ and $\widetilde{S}$ share the
values of $\alpha$, $\calpha$. $\Galpha$ and $\cvartheta$ (though
not of $\vartheta$).

\bigskip
As yet, we do not have many illuminating examples of non-commutative
graphs, nor can we offer applications to classical graph theory.
Instead, we close with highlighting several questions motivated
by the above definitions.

\begin{itemize}
  \item Algorithmic consequences: Non-commutative graph
  isomorphism is at least as hard as classical graph isomorphism,
  but are they of the same order? Similarly, graph non-isomorphism
  has efficient interactive proofs, does this extend to non-commutative
  graphs? Finally, induced substructures such as independent sets
  are NP-complete for classical graphs, and QMA-complete for non-commutative
  graphs (again for independent sets) -- but is it still QMA-complete
  for quantum independent sets? Or for entanglement-assisted
  independent sets? An interesting question in particular is, whether
  one can put a priori bounds on the dimension of the entangled state referred
  to in the definitions for $\calpha$ and $\Galpha$.

  \item For classical graphs on $n$ vertices, the largest known ratio between
  independence number and Lov\'{a}sz function occurs for random
  graphs and is $\Omega(\sqrt{n}/{\log n})$, which is conjectured
  to be maximal. What is the largest value of $\cvartheta(S)/\calpha(S)$
  for non-commutative graphs? (Our example $S=\Delta^\perp$
  in section~\ref{sec:end} shows a lower bound of $|A|/2$.)
  
  \item Random graphs are a powerful tool in combinatorics; what would
  be the natural non-commutative random graphs? The simplest one
  can think of is to fix the dimension $D$ of a subspace $S=S^\dagger < \mathcal{L}(\CC^n)$
  containing $\1$, and to choose it uniformly at random according to
  the Haar-induced measure on the Grassmannian
  (very much like what is done in~\cite{CCH-0}). What are the expected
  values of clique and independence numbers, and of 
  our $\cvartheta$ as functions of $n$ and $D$?

  \item The bipartite graphs $Z < \mathcal{L}(A\rar B)$ play a central
  role in the zero-error capacity of classical channels assisted by feedback
  or non-signalling correlation, as we have mentioned. Does this extend to
  quantum channels in the appropriate sense? For this, one first has
  to confirm that the classical noiseless feedback-assisted
  zero-error capacity, $C_{0F}(\mathcal{N})$, can be expressed in terms
  of $Z$ alone. This is indeed possible, even when the feedback is
  allowed to be an arbitrary quantum message after each channel use.
  We are currently exploring fractional packing/covering numbers for
  non-commutative bipartite graphs, with the motivation of extending
  Shannon's zero-error capacity theory to quantum channels with
  feedback.
%  Potential connections to perfect matchings and K\"onig-Hall
%  marriage theorems?

  \item There are many other graph notions we didn't generalise yet:
  Perhaps the most interesting ones are
  chromatic number and perfectness of a graph. 
  Is there a Laplacian operator with distinguished properties
  in each non-commutative graph?
  Finally, is there a good notion of edge contraction which would 
  lead to a theory of graph minors?
  
\end{itemize}

\begin{remark}
A final comment on the definition $\cvartheta(S)$:
There, it would seem more natural to consider the subspace
orthogonal with respect to the Hilbert-module inner product
$\langle \cdot,\cdot \rangle _E$ (in particular excluding the
entire diagonal $S_0$ from $S^{(\perp_E)}$). This highlights
the dependence of the notion of orthogonality on
the inner product chosen. In our definition of the independence
numbers -- and then again when we defined $\vartheta$, $\cvartheta$ -- 
we relied on the underlying Hilbert space structure,
which led us to consider the Hilbert-Schmidt inner product on $\mathcal{L}(A)$,
and more generally the conditional expectations with tracial
$\zeta_j$.
This seems to suggest that there are privileged conditional
expectations to define the Hilbert module. We leave an
investigation of this issue to future explorations of 
non-commutative graphs.
\end{remark}

%However, at least as long as we looked at \emph{generalised}
%independence numbers, these were found to be invariant under
%going to a congruent non-commutative graph. It is quite pleasing
%to see that the freedom of congruence is precisely the freedom
%of choosing different inner products
%$\langle X,Y \rangle = (\tr X^\dagger \sqrt{\rho} Y \sqrt{\rho}) \1$
%for the Hilbert module.

\acknowledgments
It is a pleasure to thank many people for discussions and feedback
on the present work, including Toby Cubitt, Debbie Leung, Will Matthews, 
Ashley Montanaro, Tomasz Paterek, Marcin Paw\l{}owski, 
and Aram Harrow.

While completing this paper, we learned
of a direct proof by Salman Beigi~\cite{salman:theta} that the
entanglement-assisted independence number of a classical channel
(and hence a classical graph) is bounded by Lov\'{a}sz' $\vartheta$
(Corollary~\ref{cor:C0-E-theta}).
We are grateful to him for sharing his manuscript with us prior to
publication.

RD is partly supported by QCIS, University of  Technology, Sydney, 
and the NSF of China (Grant Nos.~60736011 and 60702080).
SS is supported by a Newton International Fellowship.
AW is supported by the European Commission, the U.K. EPSRC, the Royal Society
and a Philip Leverhulme Prize. The Centre for Quantum Technologies is funded by 
the Singapore Ministry of Education and the National Research Foundation as 
part of the Research Centres of Excellence programme.

\end{document}